\documentclass[letterpaper]{article} 
\usepackage{aaai23}  
\usepackage{times}  
\usepackage{helvet}  
\usepackage{courier}  
\usepackage[hyphens]{url}  
\usepackage{graphicx} 
\usepackage{natbib}  
\usepackage{caption} 
\usepackage{algorithm}
\usepackage[noend]{algpseudocode}

\def\onlyBody{0}
\if\onlyBody1
\fi

\usepackage{newfloat}
\usepackage{listings}
\DeclareCaptionStyle{ruled}{labelfont=normalfont,labelsep=colon,strut=off} 
\floatstyle{ruled}
\newfloat{listing}{tb}{lst}{}
\floatname{listing}{Listing}
\newcommand{\emdash}{\,---\,}

\title{
Representation with
Incomplete Votes}
\author{
    Daniel Halpern,
    Gregory Kehne,
    Ariel D. Procaccia,
    Jamie Tucker-Foltz
    {\normalfont and} Manuel W\"uthrich\footnote{Manuel W\"uthrich was partially funded by the Swiss National Science Foundation (SNSF).}
}
\affiliations{
    Harvard University
}

\usepackage{booktabs}
\usepackage[latin9]{inputenc}
\usepackage{color}
\usepackage{babel}
\usepackage{amsmath}
\usepackage{amsthm}
\usepackage{amssymb}
\usepackage{xcolor}
\definecolor{blueish}{RGB}{103, 135, 176}
\definecolor{reddish}{RGB}{ 205, 102, 7}

\makeatletter
\usepackage{thm-restate}
\makeatother

\usepackage{cleveref}

\newtheorem{thm}{Theorem}[section]

\newtheorem{lem}[thm]{Lemma}
\newtheorem{defn}{Definition}[section]

\usepackage{amsmath}
\usepackage{amssymb}
\usepackage{graphicx}
\usepackage{caption}
\usepackage{amsthm}
\usepackage{color}
\definecolor{green}{rgb}{0,0.5977,0}

\newcommand{\rr}{\mathbb R}

\newcommand{\abs}[1]{\left|{#1}\right|}

\newcommand{\suchthat}{\ | \ }

\newcommand{\tth}{^\text{th}}
\newcommand{\genseq}[3]{{#1}_1 {#3} {#1}_2 {#3} \dots {#3} {#1}_{#2}}
\newcommand{\seq}[2]{\genseq{#1}{#2}{,}}

\usepackage{enumitem}

\usepackage{mathtools}
\DeclarePairedDelimiter{\set}{\{}{\}}
\DeclarePairedDelimiter{\ceil}{\lceil}{\rceil}

\newcommand{\pavsc}{\textsc{pav-sc}}
\newcommand{\pavd}{\Delta}
\DeclareMathOperator*{\argmax}{arg\,max}

\begin{document}

\maketitle
\begin{abstract}
    Platforms for online civic participation rely heavily on methods for condensing thousands of comments into a relevant handful, based on whether participants agree or disagree with them. These methods should guarantee fair representation of the participants, as their outcomes may affect the health of the conversation and inform impactful downstream decisions. To that end, we draw on the literature on approval-based committee elections. Our setting is novel in that the approval votes are incomplete since participants will typically not vote on all comments. We prove that this complication renders non-adaptive algorithms impractical in terms of the amount of information they must gather. Therefore, we develop an adaptive algorithm that uses information more efficiently by presenting incoming participants with statements that appear promising based on votes by previous participants. We prove that this method satisfies commonly used notions of fair representation, even when participants only vote on a small fraction of comments. Finally, an empirical evaluation using real data shows that the proposed algorithm provides representative outcomes in practice.
\end{abstract}

\section{Introduction}\label{secIntro}

A recent surge of interest in empowering citizens through online civic participation has spurred the development of a number of platforms \cite{Salganik2015-xa, Ito2020-vs, Shibata2019-dq, Fishkin2018-zi, Aragon2017-mp,
    Iandoli2009-do}.
A particularly successful example is \emph{Polis} \cite{SBES+21},\footnote{\texttt{https://pol.is}} an open-source ``system for gathering, analyzing, and understanding what large groups of people think in their own words.'' It has been widely used by local and national government agencies around the world. Most notably, it is the basis of vTaiwan, a system commissioned by the government of Taiwan, whose participatory process\emdash involving thousands of ordinary citizens\emdash
has led to new regulation of ride-sharing services and financial technology.
A similar (albeit commercial) system called \emph{Remesh}\footnote{\texttt{https://www.remesh.ai}} allows users to ``save resources by democratizing insights in live, flexible conversations with up to 1,000 people at the same time.''\footnote{See also \texttt{consider.it}, \texttt{citizens.is}, \texttt{ make.org}, \texttt{kialo.com}.}

The key idea underlying both systems is simple and broadly applicable: Participants can submit free-text comments about the discussion topic at hand and choose to agree or disagree with others' comments presented to them by the platform. An essential part of the process is the aggregation of these opinions toward an ``understanding of what large groups of people think.'' Polis, for instance, displays a list of comments that received the most support among participants to whom they were shown. But this aggregation method may fail to represent minority groups, even those that are very large: if 51\% of participants agree with one set of comments, while 49\% of participants agree with another set of comments, only comments from the first set will appear on this list. Polis has recognized this problem and sought to mitigate it by employing a second, more elaborate procedure~\cite{SBES+21}.\footnote{sThe idea is to find clusters of participants with similar opinions and then ensure that each cluster is represented by comments that distinguish it from the others.
} While this procedure has produced interesting results in practice, it does not guarantee summarizations that are representative of the discussion in a rigorous sense.

In this paper, we reexamine opinion aggregation in systems such as Polis and Remesh through the lens of computational social choice~\cite{BCEL+16}. We observe that \emph{selecting a subset of comments based on agreements and disagreements is equivalent to electing a committee based on approval votes}. From this viewpoint, the primary aggregation method used by Polis corresponds to classical approval voting (AV). There is substantial work\emdash starting with the paper of \citet{ABCE+17}\emdash on \emph{approval-based committee elections}
that seeks to avoid the shortcomings of AV by guaranteeing that the selected committee satisfies fairness notions.
To define one such notion (which is not satisfied by AV), note that if the size of the committee is $k$ and the number of voters is $n$, a subset of $n/k$ voters is large enough to demand a seat on the committee if they agree on at least one candidate. This intuition is captured by a property called \emph{justified representation (JR)}, which guarantees that every such subset of voters has an approved candidate on the committee.

There is a major gap, however, between the literature on approval-based committee elections and the reality of systems like Polis and Remesh: these systems only have access to partial votes. For example, in the discussion facilitated by Polis around ridesharing regulation in Taiwan, 197 comments were submitted, but each participant only cast a vote on 10.57 comments on average\emdash roughly 5\% of all comments. Our main conceptual insight is that we can overcome the partial-information gap via statistical estimation and adaptive querying (i.e., by deciding which comments to show to incoming users based on previous votes).\footnote{There is a body of work in computational social choice related to incomplete votes.
    For instance, some papers aim to find winning committees, given incomplete approval votes, or to fill in the missing votes, given knowledge about the domain of approval profiles~\cite{Imber2022-bt, Terzopoulou_Karpov_Obraztsova_2021, Zhou2019-sp}. However, these papers are primarily concerned with the computational complexity of these problems, while we focus on information-theoretic questions.
    There is also related work that studies the problem of determining the winner given only partial rankings~\citep{XC11,FO14}, but this setting is mathematically different from ours.
    Furthermore, prior work does not consider the adaptive setting, where we query voters sequentially and decide on the next question based on previous votes.
}

\paragraph{Our approach and results.}\label{sec:approach}
In our model, each voter (user) can be asked to express their opinion (approval/disapproval) about at most $t$ candidates (comments). More formally, a \emph{query} asks a randomly-chosen voter for their approval votes on a subset of candidates $S$ of size $|S|\leq t$. Note that this query model is consistent with how Polis works, where participants express their agreement or disagreement with the comments shown to them by the system.
We can view the response to such a query, i.e., the approval votes of a single voter, as noisy information about the profile of the entire population of voters (restricted to these $t$ candidates). Therefore, we refer to these real-world queries as \emph{noisy queries}.

Before we discuss this setting, we consider a simplified setting in \Cref{secExactQueries}, where queries yield the profile of the entire population of voters on the $t$ candidates in the query. While such \emph{exact queries} are not realistic, they provide an abstraction that is easier to study and allows us to derive lower bounds on the number of queries required to achieve JR (which apply also to the noisy-query setting, since it is strictly harder).
We start by studying the required number of queries of \emph{non-adaptive} algorithms, which decide on their queries before any votes are cast. While non-adaptive algorithms may be preferable in some cases (e.g., because no voter can influence what alternatives are shown to other voters or because computation can be performed offline), we show that they are impractical because they must ask at least $\Omega(m^{11})$ queries (and hence voters) to achieve JR, where $m$ is the number of candidates.

Therefore, we focus on \emph{adaptive} algorithms in the rest of the paper. In \Cref{subAdaptive} we adapt a local-search algorithm of \citet{AEHL+18} to the case of exact queries and show that it can achieve JR (and even stronger properties) with $\mathcal{O}(mk^2\log k)$ queries.

In \Cref{secNoisyQueries}, we move on to the realistic, noisy-query model, where a query corresponds to a single voter. Since we need to estimate the answer to each exact query using multiple noisy queries to control uncertainty, the query complexity of the adaptive algorithm for the same guarantees increases to $\mathcal{O}\left(mk^6\log k\log m\right)$. By applying martingale theory, we develop an extension of this algorithm that allows the reuse of votes in a statistically sound way.

In \Cref{secEmpirical} we show empirically (on real datasets from Polis and Reddit) that this extension allows us to find committees satisfying (approximate) JR (and stronger properties) despite access to little information (i.e., few voters, each voting on only a small fraction of the comments).

\section{Preliminaries}\label{secModel}

We begin by introducing the standard approval-based committee selection setting~\citep{ABCE+17}. For $s \in \mathbb{N}$, we use the notation $[s] = \set{1, \ldots, s}$. We have a set $N = [n]$ of $n$ voters and a set $C$ of $m$ candidates. Each voter $i \in N$ approves a set of candidates $A_i \subseteq C$. We refer to the vector $\mathbf{A} = (A_1, \ldots, A_n)$ as an \emph{approval profile}. The goal is to choose a \emph{committee} $W \subseteq C$ of size $k \le m$. The value $k$ is called the \emph{target committee size}. We refer to an algorithm that takes as input the profile and candidates and outputs a committee of size $k$ as a \emph{$k$-committee-selection algorithm}.

\medskip
\noindent\textbf{Notions of representation.}
We say that a group of voters $V \subseteq N$ is \emph{$\ell$-large} if $|V| \ge \ell \cdot \frac{n}{k}$; $V$ is \emph{$\ell$-cohesive} if $|\bigcap_{i \in V} A_i| \ge \ell$. \citet{ABCE+17} introduced the following two fairness notions:

\begin{defn}[Justified Representation (JR)]
    A committee $W$ provides JR if for every $1$-large, $1$-cohesive group of voters $V$, there exists a voter $i \in V$ who approves a member of $W$, i.e., $|A_i \cap W| \ge 1$.
\end{defn}
\begin{defn}[Extended Justified Representation (EJR)]
    A committee $W$ provides EJR if for every $\ell \in [k]$ and every $\ell$-large, $\ell$-cohesive group of voters $V$, there exists a voter $i \in V$ who approves at least $\ell$ members of $W$, i.e., $|A_i \cap W| \ge \ell$.
\end{defn}
We also study the following approximate version of EJR:
\begin{defn}[$\alpha$-Extended Justified Representation ($\alpha$-EJR)]
    A committee $W$ provides $\alpha$-EJR if for every $\ell \in [k]$ and every $\frac{\ell}{\alpha}$-large, $\ell$-cohesive group of voters $V$, there exists a voter $i \in V$ who approves at least $\ell$ members of $W$, i.e., $|A_i \cap W| \ge \ell$.
\end{defn}
\citet{Sanchez-Fernandez2016-tk} proposed another notion of representation
called the \emph{average satisfaction} of a group of voters $V$ for
a committee $W$, defined as $\text{avs}_{W}(V)=\frac{1}{|V|}\sum_{i\in V}|A_{i}\cap W|$.
Related to this quantity, we define the following property:
\begin{defn}[$\alpha$-Optimal Average Satisfaction ($\alpha$-OAS)]
    A committee $W$ provides \emph{$\alpha$-OAS} if for every $\lambda\in[0,k]$
    and every $\frac{\lambda+1}{\alpha}$-large, $\lambda$-cohesive group
    of voters $V$, we have $\text{avs}_{W}(V)\ge\lambda$.
\end{defn}
This property measures how close a committee is to the maximum average satisfaction that can be guaranteed to hold for all elections. To see this, note that the condition above is equivalent to the following condition: for every $\ell\in[\frac{1}{\alpha},\frac{k+1}{\alpha}]$
and every $\ell$-large, $(\alpha\ell-1)$-cohesive group of voters
$V$, we have $\text{avs}_{W}(V)\ge\alpha\ell-1$. This implies
a \emph{proportionality guarantee} \cite{Skowron2021-ys} of $g(\ell,k)=\alpha\ell-1$.
Since there is no selection rule that satisfies a proportionality
guarantee with $g(\ell,k)>\ell-1$ for all elections \cite{AEHL+18, Skowron2021-ys}, $\alpha=1$ is the best we can hope for, so we refer to $1$-OAS simply as OAS.

\medskip
\noindent\textbf{Proportional approval voting.}
\emph{Proportional Approval Voting
    (PAV)} is a widely-studied committee selection algorithm: given an approval profile $\mathbf{A}$ and a committee size $k$, it returns a committee $W$ of size $k$ maximizing the \emph{PAV score},
defined as
\begin{equation*}
    \pavsc(W):=\frac{1}{n}\sum_{i\in N}\sum_{j=1}^{|A_{i}\cap W|}\frac{1}{j}.
\end{equation*}
PAV satisfies EJR and OAS \citep{Sanchez-Fernandez2016-tk,AEHL+18},
but is NP-hard to compute~\citep{Aziz2015-ce}. Consequently, \citet{AEHL+18}
propose a local search approximation of PAV (LS-PAV), which continues to satisfy EJR and OAS, but, unlike PAV, runs in polynomial time.
As we shall see, LS-PAV is a useful basis for algorithms in our query model.

\section{Exact Queries}\label{secExactQueries}
In the exact-query setting,
the response $R$ to a query $Q$ consists of a proportion $p_S$ for every subset $S \subseteq Q$, where $p_S$ is the proportion of voters who only approve the candidates in $S$ among the queried candidates $Q$, i.e.,
\[
    p_S := \frac{1}{n} \: \sum_{i \in N} \mathbb{I}[A_i \cap Q = S],
\]
where $\mathbb{I}$ is the indicator function.
We refer to an algorithm that makes queries of size $t$, receives this type of response, and outputs a committee of size $k$ as a \emph{$(k, t)$-committee selection algorithm with exact queries}. We say an algorithm is \emph{adaptive} if the queries it chooses depend on responses from previous queries. Note that we allow all of our algorithms to be randomized. In the following, we ask how many queries are needed to guarantee the notions of representation introduced in \Cref{secModel}.

\subsection{Nonadaptive Algorithms}\label{subNonadaptive}

In this section, we think of $m$ as large (many comments will be submitted to the system), while we think of $k$ and $t$ as small constants (since we wish to select only a few comments and voters have limited time).
Since we are primarily interested in \emph{lower} bounds on the query complexity of non-adaptive algorithms, we consider only JR, the weakest fairness criterion.

An initial observation is that, if $t\ge k$, JR can always be guaranteed with $O(m^k)$ queries, as simply querying every set of $k$ candidates provides all the information necessary to run PAV. For $k = 1$, this bound is tight, as voters could all approve only a single candidate, which will take a linear number of queries to find. Our first result is a tight quadratic lower bound for $k = 2$.

\begin{restatable}{thm}{ExactLowerSquare}\label{thmExactLowerBoundN2}
    For any constants $k$ and $t$ such that $k \geq 2$, and any $\varepsilon > 0$, any non-adaptive $(k, t)$-committee selection algorithm that makes fewer than $\Omega(m^2)$ queries satisfies JR with probability at most $\varepsilon$.
\end{restatable}

This result provides a separation between the non-adaptive and the adaptive settings: In Section~\ref{subAdaptive}, we discuss an adaptive $(k, t)$-committee selection algorithm guaranteeing JR with only $O(m)$ queries for any $k$ and $t$ such that $k < t$.

Theorem~\ref{thmExactLowerBoundN2} follows from a more general result that we present formally in Appendix~\ref{appComputationalSearch}. Here, we illustrate the argument by considering the special case where $t = k = 2$ and $\varepsilon = \frac{5}{6}$: Consider an adversary
that picks a random set of 3 candidates, call them 1, 2, and 3, and answers queries according to the approval matrix visualized in Figure~\ref{figMergedCounterexamples}(a): half of the voters approve only candidate 1, and the other half of the voters approve only candidates 2 and 3. To satisfy JR, the algorithm needs to include candidate 1 in the committee. However, if the algorithm never queries $\{1, 2\}$, $\{1, 3\}$, or $\{2, 3\}$, it receives no information that can distinguish candidates 1, 2, and 3 from each other, so it can do no better than selecting a random pair from these three candidates, which will succeed with probability $\frac{2}{3}$.
This problematic case will occur frequently if the number of queries is not very large, say $\frac{1}{18} \cdot \binom{m}{2}$: Since there are $\binom{m}{2}$ pairs of candidates, the probability that the algorithm queries any randomly selected pair of candidates is at most $\frac{1}{18}$. By the union bound, the probability that the algorithm queries any of $\{1, 2\}$, $\{1, 3\}$, or $\{2, 3\}$ is at most $3\cdot \frac{1}{18} = \frac{1}{6}$. To summarize, for the algorithm to succeed, it either needs to get lucky during the querying phase, which happens with probability at most $\frac{1}{6}$, or during the selection phase, which happens with probability at most $\frac{2}{3}$. By the union bound, the algorithm succeeds with probability at most $\frac{1}{6} + \frac{2}{3} = \frac{5}{6}$.

A natural follow-up question is whether the $O(m^k)$ upper bound is tight for larger $k$. Interestingly, this is \emph{not} the case for $k \geq 3$, as we prove in Appendix~\ref{appT2K3JR}:
\begin{restatable}{thm}{NonadaptiveJR}\label{thmNonadaptiveAlgorithmKEqualsTPlust2}
    For any $t \geq \frac{2}{3} k$, there exists a $(k, t)$-committee selection algorithm guaranteeing JR with $O(m^{t})$ exact queries.
\end{restatable}

However, the exponent does have a dependence on $k$. In particular, we find that guaranteeing JR requires $\Omega(m^3)$ queries starting at $k = 6$. The adversary employs an analogous strategy, now picking 7 random candidates and imposing the approval matrix depicted in Figure~\ref{figMergedCounterexamples}(b). Satisfying JR requires that the algorithm include candidate 1, which is indistinguishable from the other six candidates unless the algorithm makes $\Omega(m^3)$ queries, since every candidate is approved by $\frac{6}{18}$ of the voters and every pair of candidates is approved by $\frac{2}{18}$ of the voters.

\begin{figure}[t]
    \begin{center}\includegraphics[scale = .3]{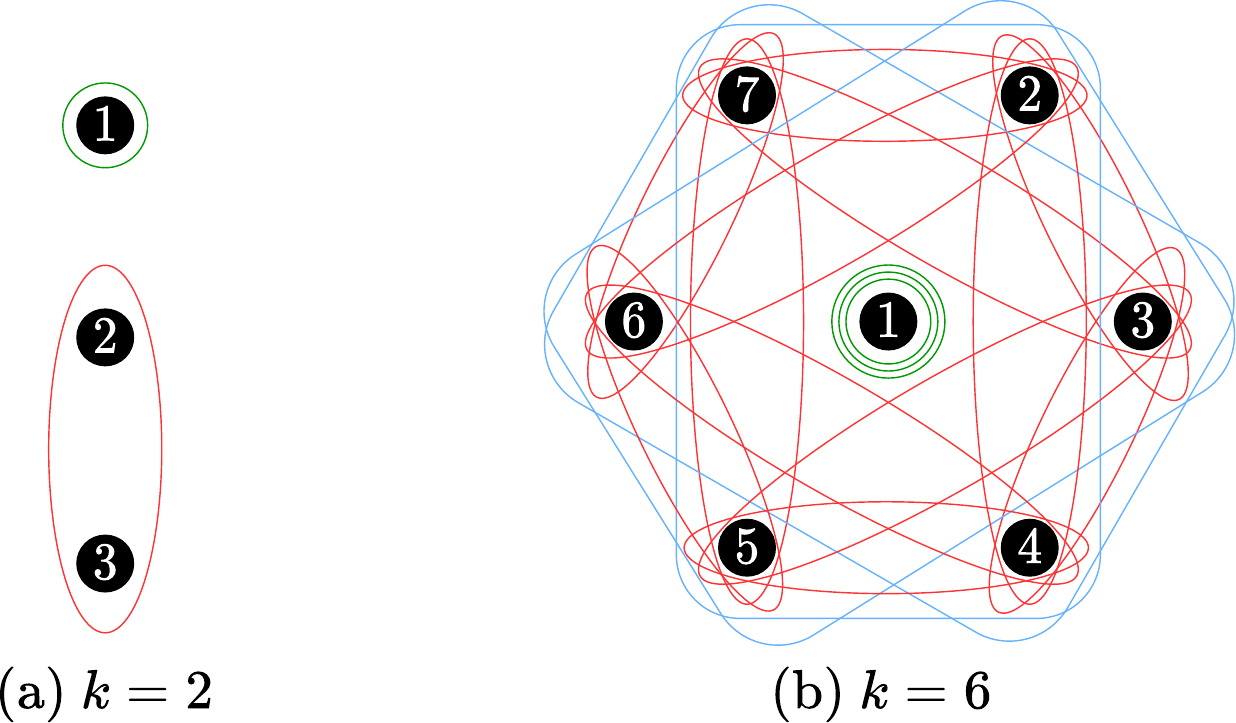}
        \caption{
            \label{figMergedCounterexamples} Adversarial approval matrices. Each region represents a disjoint, equally-sized set of voters who approve only the candidates within the region. In (a), queries of size $t \geq 2$ are needed to distinguish the candidates; in (b), we need $t \geq 3$.
        }
    \end{center}
\end{figure}

In Appendix~\ref{appComputationalSearch}, we describe a computational search we conducted to find similar instances for larger values of $k$. The best lower bound obtained is as follows.

\begin{restatable}{thm}{ExactLowerNEleven}\label{thmExactLowerBoundN11}
    For any $\varepsilon > 0$, there exists a target committee size $k$ with $k = \Theta(\log 1/\varepsilon)$ such that for all $t$, any non-adaptive $(k, t)$-committee selection algorithm with exact queries that makes fewer than $\Omega(m^{11})$ queries satisfies JR with probability at most $\varepsilon$.
\end{restatable}

This theorem closes the book on the (im)practicality of non-adaptive committee selection algorithms. We therefore turn our attention to adaptive algorithms.

\subsection{An Efficient Adaptive Algorithm}\label{subAdaptive}

In this section, we propose an adaptive algorithm based on LS-PAV~\citep{AEHL+18}, and we show that it achieves EJR and OAS with a practically-feasible number of queries.

For convenience, we introduce the following notation:
For a committee $W$ and candidates $c \in W$ and $c' \notin W$, let
$$\pavd(W, c', c) := \pavsc(W \cup \set{c'} \setminus \set{c}) - \pavsc(W)$$
denote the difference in PAV score obtained by replacing $c$ with $c'$ in $W$. Additionally, let
$$\pavd(W, c) := \pavsc(W \cup \set{c}) - \pavsc(W)$$
denote the marginal gain in PAV score by adding $c$ to $W$.

LS-PAV starts with an arbitrary committee $W$ and repeatedly replaces
a committee member $c\in W$ with a candidate $c'\notin W$, provided the improvement
to the PAV score satisfies $\pavd(W,c',c)\ge\frac{1}{k^{2}}$. \citet{AEHL+18}
show that after at most $\mathcal{O}(k^{2}\log k)$ swaps, no such swap pairs $c, c'$ remain, at which point $W$ satisfies OAS and EJR.

We first observe that LS-PAV can be implemented using exact queries: For any set of candidates $S$, $\pavsc(S)$ can be computed using any query $Q \supseteq S$, as it is sufficient to know the proportion of voters that approve each subset of $S$. Hence, for any $W$, $c \in W$, and $c' \notin W$, $\pavd(W, c', c)$ can be computed using a query $Q$ that contains both $W$ and $c'$. Using $\ceil*{ \frac{m-k}{t-k}}$
queries of size $t$, we can cover all $m-k$ candidates that are not in $W$, which leads to an overall query complexity of $\mathcal{O}\left(m k^{2}\log k\right)$.

We next present a version of LS-PAV, which we call $\alpha$-PAV (\Cref{alg:a-pav}),
that has the same query complexity as LS-PAV for finding a committee
that satisfies EJR and OAS, but lower query complexity
for approximate ($\alpha<1$) $\alpha$-EJR and $\alpha$-OAS.
\begin{algorithm}[tb]
    \caption{$(k,t)$-$\alpha$-PAV}
    \label{alg:a-pav}
    \begin{algorithmic}[1] 
        \State Choose $W \in \binom{C}{k}$, $c \in W$, and $c' \notin W$ arbitrarily
        \State $\gamma \gets \infty$
        \While{$\gamma \geq \frac{1}{\alpha k}$}
        \State $W \gets W\cup \set{c'} \setminus \set{c}$
        \State Choose $\mathcal{Q} = \{Q_i\}_i$, with $|Q_i|=t$, s.t. $W \subseteq\bigcap\mathcal{Q}$ \par
        \par\hspace{-0.1in} and $C \subseteq \bigcup\mathcal{Q}$
        \State $c' \gets \argmax_{x\notin W}\pavd(W,x)$ \Comment{(using $\mathcal{Q}$)}
        \State $c \gets \argmax_{x\in W}\pavd(W,c',x)$ \Comment{(using $\mathcal{Q}$)}
        \State $\gamma \gets \pavd(W,c')$
        \EndWhile
        \State \Return $W$
    \end{algorithmic}
\end{algorithm}
Besides introducing the approximation parameter $\alpha$, we make two other modifications to LS-PAV:
First, \Cref{alg:a-pav} terminates when there is no candidate $c'$ such that $\pavd(W,c')\ge \frac{1}{k}$ (for $\alpha=1$), while LS-PAV terminates when there is no pair $c, c'$ such that $\pavd(W,c', c)\ge \frac{1}{k^2}$.  As we shall see in \Cref{lemPotentialIncrease}, the termination condition of \Cref{alg:a-pav} is weaker than that of LS-PAV, implying that it may terminate earlier.
Second, instead of considering all possible swaps $c,c'$, we only consider adding the candidate
$c'$ with the largest $\pavd(W,c')$. This modification makes the algorithm slightly simpler and more computationally efficient
(by a factor of $k$).

\begin{restatable}{thm}{GlsPavGuarantees}\label{thmExactQueryUpper}
    \label{thm:a-pav_guarantees}
    For any $m \ge t > k$,
    \Cref{alg:a-pav} yields a committee satisfying
    $\alpha$-OAS and $\alpha$-EJR while making at most
    \begin{equation*}
        \left\lceil \frac{m-k}{t-k}\right\rceil \: \frac{\alpha k^{2}}{(1-\alpha)k+1} \: H_k
    \end{equation*}
    queries, where $H_k$ is the $k\textsuperscript{th}$ harmonic number.
    For ${\alpha = 1}$, this leads to a query complexity of $\mathcal{O}\left(m k^{2}\log k\right)$ while for any fixed $\alpha < 1$, this leads to a query complexity of $ \mathcal{O}\left(m k\log k\right)$.
\end{restatable}

The proof of \Cref{thmExactQueryUpper} essentially follows from the following two lemmas, the first of which uses the notation
\[
    \pavd^*(W) := \max_{c\in C} \pavd (W, c).
\]

\begin{restatable}{lem}{lemLowDeltaImpliesStuff}\label{lemLowDeltaImpliesStuff}
    If a committee $W$ satisfies $\pavd^{*}(W)<\frac{1}{\alpha k}$, then $W$ satisfies $\alpha$-EJR and $\alpha$-OAS.
\end{restatable}

\begin{restatable}{lem}{lemPotentialIncrease}\label{lemPotentialIncrease}
    For any committee $W$ and ${c \notin W}$, we have that $\max_{x \in W}\pavd(W,c,x)\ge\frac{(k+1)\pavd(W,c)-1}{k}$.
    In particular, if $\pavd(W,c)\ge\frac{1}{\alpha k}$,
    then $\max_{x \in W}\pavd(W,c,x)\ge\frac{(1-\alpha)k+1}{\alpha k^{2}}$.
\end{restatable}

\Cref{lemLowDeltaImpliesStuff} guarantees that when \Cref{alg:a-pav} terminates the desired fairness properties are satisfied. \Cref{lemPotentialIncrease} establishes that the PAV score increases over the algorithm's run. This bounds the number of swaps it performs since $\pavsc(W)$ is at most $H_k$.

\Cref{lemLowDeltaImpliesStuff} is a generalization of the lower bound from Lemma 1 of \citet{Skowron2021-ys}. This generalization is useful because it states that to establish EJR and OAS of any given committee $W$ (no matter how it is derived), it is sufficient to prove that $\pavd^{*}(W)$ is small; hence it can be used as a certificate of satisfaction. In \Cref{sec:pav_proofs}, we show that standard PAV and LS-PAV satisfy $\pavd^{*}(W)<\frac{n}{k}$, which is noteworthy in that it provides a simple proof of the known result that they satisfy EJR and OAS.

We observe that, for exact queries, an $\alpha$-approximation with $\alpha<1$ improves the query complexity by a factor of $k$. In the next section, we will see that such an approximation yields an even larger improvement in query complexity for noisy queries, as it also reduces the accuracy with which we need to estimate $\pavd (W, c', c)$.

\section{Noisy Queries}\label{secNoisyQueries}

We now turn to a query model that includes the noise we abstracted away in \Cref{secExactQueries}. In order to represent voters arriving to the platform one-by-one, we assume that each time the algorithm performs a query $Q \subseteq C$ a voter $i \in N$ is selected independently and uniformly at random%
\footnote{Note that a voter-profile $A_i$ may be queried more than once during the run of the algorithm because we sample \emph{with replacement}. This model simplifies the statistical analysis and has a natural interpretation: Rather than thinking of a finite population of voters, we draw samples from an underlying population distribution where each profile $A_1,..,A_n$ has the same frequency (probability). Furthermore, our model approaches sampling without replacement if the size of the underlying population $n$ is large compared to the number of queried voters, hence both models are qualitatively interchangeable.}%
and then the algorithm observes their votes on the queried candidates $Q \cap A_i$.
We refer to an algorithm that performs queries of size $t$, receives as a response the votes of a single voter, and outputs a committee of size $k$ as a \emph{$(k, t)$-committee selection algorithm with noisy queries}.

To see the connection between this query model and the previous one, note that an algorithm with noisy queries can approximate an exact query $Q$ by estimating the values of $p_S$ by taking the empirical proportion of repeated samples.
By standard sample complexity bounds, using $\Theta\left(\log(2^t/\delta)/\varepsilon^2\right)$ queries, a noisy-query algorithm could guarantee $\pm \varepsilon$ estimates of $p_S$ for all $S\subseteq Q$ with probability $1 - \delta$. Hence if an exact-query algorithm requires no more than $\text{poly}(m)$ queries with additive $\varepsilon$ error, then it can be implemented using a factor of $\Theta(\log m)$ more noisy queries and yield a correct result with probability $1 - \delta$.
What's more, this $\log$ factor is in some cases necessary when moving from the exact-query to the noisy-query setting. In \Cref{appNoisyVsExact}, we demonstrate instances for which a non-adaptive exact-query algorithm needs only $\Theta(m)$ queries, while in order to be correct with any fixed probability $\delta$, a non-adaptive noisy-query algorithm requires $\Omega(m \log m)$ queries.

Conversely, notice that one can use exact queries to simulate noisy queries. Indeed, $p_S$ is exactly the probability that an incoming voter will vote yes on candidates $S$ and no on $Q\setminus S$ in response to a query $Q$. An algorithm with access to exact query values can simply sample a voter response and feed it to a noisy-query algorithm. Therefore, the lower bounds on the query complexity of exact-query, non-adaptive algorithms, in particular \Cref{thmExactLowerBoundN11}, apply to noisy-query, non-adaptive algorithms as well.
As the number of candidates becomes large, adaptivity is therefore necessary to attain theoretical guarantees\emdash mirroring the approach of online platforms in practice.

A natural starting point is the exact-query adaptive algorithm, namely \Cref{alg:a-pav}. Indeed, it can be adapted to the noisy setting by replacing exact queries with a sufficient number of noisy queries, $\ell$, to obtain high-probability bounds on $\pavd$, yielding \Cref{alg:noisy-a-pav}.

\begin{algorithm}[tb]
    \caption{$(k,t)$-noisy-$\alpha$-PAV}
    \label{alg:noisy-a-pav}
    \begin{algorithmic}[1]
        \State $\ell \gets \ceil*{288  \left( \frac{\alpha k ^2}{(1 - \alpha)k + 1} \right)^2 \log \left(\frac{8mk^4}{\delta} \right)}$
        \State Choose $W \in \binom{C}{k}$, $c \in W$, and $c' \notin W$ arbitrarily
        \State $\gamma \gets \infty$
        \While{$\gamma \geq 1/(\alpha k) - ((1 - \alpha)k + 1)/(12\alpha k^2)$}
        \State $W \gets W\cup \set{c'} \setminus \set{c}$
        \State Choose $\mathcal{Q} = \{Q_i\}_i$, with $|Q_i|=t$, such that \par
        \par\hspace{-0.1in} ${W\subseteq\bigcap\mathcal{Q}}$ and $C \subseteq \bigcup\mathcal{Q}$
        \State Ask each query $Q \in \mathcal{Q}$ to $\ell$ new voters
        \State $\hat{\Delta}(W, x) \gets$ estimate of $\Delta(W, x)$ using $\ell$ voters
        \par\hspace{-0.1in}  from query $Q$ containing $W \cup \set{x}$ \Comment{$\forall x \notin W$}
        \State $\hat{\Delta}(W, x, y) \gets$ estimate of $\Delta(W, x, y)$ using $\ell$ voters
        \par \hspace{-.1in} from $Q$ containing $W \cup \set{x}$ \Comment{$\forall x \notin W, \forall y \in W$}
        \State $c'\gets \argmax_{x\notin W}\hat{\pavd}(W, x)$
        \State $c \gets \argmax_{x\in W}\hat{\pavd}(W, c',x)$
        \State $\gamma \gets \hat{\pavd}(W, c')$
        \EndWhile
        \State \Return $W$
    \end{algorithmic}
\end{algorithm}
The key is to choose $\ell$ large enough that if the termination condition is not met, i.e., we have $\hat{\pavd}(W,c')<\frac{1}{\alpha k} - \frac{(1 - \alpha)k + 1}{12\alpha k^2}$, the resulting swap is guaranteed to yield a positive improvement in the PAV-score, such that the number of steps of the algorithm is bounded. With the choice of $\ell$ in \Cref{alg:noisy-a-pav}, we obtain the following theorem, whose proof can be found in \Cref{sec:noisy-a-pav-bound}.
\begin{restatable}{thm}{noisyAPAVBound} \label{thm:noisy-a-pav-bound}
    For any $m\ge t > k$, with probability at least $1-\delta$, \Cref{alg:noisy-a-pav}
    returns a committee that satisfies $\alpha$-EJR and $\alpha$-OAS
    after querying no more than
    \begin{align*} 578H_k\left\lceil \frac{m-k}{t-k}\right\rceil \left(\frac{\alpha k^{2}}{(1-\alpha)k+1}\right)^{3}\log\left(\frac{4mk^{4}}{\delta}\right)
    \end{align*}
    voters. For any fixed $\delta > 0$, if $\alpha=1$, this leads to a query complexity of $\mathcal{O}\left(mk^6\log k\log m\right)$,
    and if $\alpha<1$, this leads to a query complextiy of $\mathcal{O}\left(mk^3\log k\log m\right)$.
\end{restatable}

While \Cref{alg:noisy-a-pav} achieves good worst-case query complexity, it may be suboptimal on certain instances because of two reasons: (i) after each swap, \Cref{alg:noisy-a-pav} discards all previous information so each candidate is reassessed from scratch, and
(ii) it presents each candidate $c\notin W$ to the same number of voters, even though it may quickly become apparent that some candidates are more promising than others.

To address issue (i), we can use all past votes to compute bounds on $\pavd$. A difficulty with this approach is that past voters may not have voted on all candidates in $W$ (which is necessary to directly estimate $\pavd(W,c)$), since they may have been queried on a different committee $W'$.
But we can nonetheless use these past votes to obtain upper and lower bounds on estimated values.
To address issue (ii), we can present promising candidates
to voters more often. Further, it is possible to perform swaps as soon as we are confident they yield an increase of the PAV-score of at least some value $\varepsilon$, rather than first querying a predetermined number of voters as in \Cref{alg:noisy-a-pav}.

These ideas are incorporated into \Cref{alg:ucb-a-pav}, called ucb-$\alpha$-PAV; see Appendix~\ref{sec:ucb-a-pav-bound} for a formal description of the algorithm and an analysis of its query complexity.

\section{Experiments}\label{secEmpirical}

Since the analysis in the theoretical sections considers worst-case approval profiles, it is possible that, in practice, we may be able to find good committees with fewer queries than required by \Cref{thm:noisy-a-pav-bound}.
We investigate this question empirically on real data from online discussions with only a few hundred voters, each voting on only a fraction of all comments.

\paragraph{Datasets.}
Polis provides open-use data from real deliberations hosted
on their platform.\footnote{\texttt{https://github.com/compdemocracy/openData}}
These include, for instance, a discussion organized by the government
of Taiwan, which led to the successful regulation of Uber.
Since participants
typically only vote on a fraction of comments, most votes are missing.
To be able to simulate the proposed adaptive algorithms, we first infer
these missing votes using a matrix factorization library, LensKit.\footnote{\texttt{https://lenskit.org}}
Importantly, we infer votes only for the purpose of the experiments; if our algorithms were executed during the discussion, they would adaptively
query users about the relevant comments without relying on any inference
method.

In most datasets, we observe several comments that are nearly universally approved.
Since these comments make achieving EJR and
OAS trivial, we remove comments approved by more than $60\%$ of participants. This step may also be appropriate in practice to gain insights into participants' opinions beyond uncontroversial issues.

The number of queried voters $L$ ranges from $87$ to $1000$ across the $13$ datasets (see \Cref{app:experiments} for details). For all datasets, we assume that each voter votes on $t=20$ comments. Since the total number of comments $m$ ranges from $31$ to $1719$ across datasets, the percentage of comments each voter votes on, $t/m$, ranges from $1\%$ to $65\%$.
For
each dataset, we run the algorithms with target committee sizes $k=5, 7, 10$. Hence, there are a total of $13\cdot 3=39$ experiments (times $10$ random seeds).

The second dataset we consider consists of Reddit discussions.\footnote{\texttt{https://www.kaggle.com/datasets/}\\\texttt{josephleake/huge-collection-of-reddit-votes}}
To obtain an interesting dataset, we combined voting data from two subreddits, r/politics and r/Conservative, which are arguably situated at opposite ends of the American political spectrum.
More details about this dataset can also be found in \Cref{app:experiments}.

\paragraph{Algorithms.}
We evaluate noisy-$\alpha$-PAV (\Cref{alg:noisy-a-pav}) and ucb-$\alpha$-PAV (\Cref{alg:ucb-a-pav}). Both query $L$ voters in random order, each of whom votes on $t=20$ comments.
To enable these algorithms to swap candidates after querying only a small number of voters, we make the following modifications:
For both \Cref{alg:noisy-a-pav} and \Cref{alg:ucb-a-pav} we treat $\ell$, the number of times we ask voters about each candidate, as a parameter. In addition, for \Cref{alg:ucb-a-pav},
we replace the numerator in the confidence intervals
$err_s$ with a parameter $\theta$.
Both $\ell$ and $\theta$ were chosen based on validation on a separate dataset, see \Cref{app:experiments} for details. We run both algorithms on all the $L$ voters, rather than terminating as soon as we can guarantee $\pavd{}^*(W)<\frac{1}{\alpha k}$ (and hence EJR and OAS).

To obtain an upper bound on the attainable performance, we execute $\alpha$-PAV (\Cref{alg:a-pav}) with access to exact queries. To obtain the best possible $\alpha$, we let \Cref{alg:a-pav} run as long as the swap increases the PAV score, i.e., $\Delta(W,c',c)>0$, instead of terminating as soon as $\Delta(W,c') < 1 / k$ (which would be sufficient to guarantee EJR and OAS).

To verify that the proposed algorithms do indeed take the complementarity of different candidates into account, we also compare against standard approval voting (AV) with access to all votes, which simply selects the $k$ candidates with the most approval votes.

\paragraph{Performance Metric.}
As a performance metric,
we use $\hat{\alpha} := \frac{1}{k\Delta^{*}(W)}$, where $W$ is the committee selected by the respective algorithm. According to \Cref{lemLowDeltaImpliesStuff}, $\alpha > \hat{\alpha}$, so this implies $\hat{\alpha}$-EJR and $\hat{\alpha}$-OAS. As discussed in \Cref{secModel}, $\alpha = 1$ is the best that can be guaranteed across all possible approval profiles.
Note that $\alpha$ may be larger than $\hat{\alpha}$, hence obtaining $\hat{\alpha}=1$ is a sufficient, but not a necessary condition for OAS and EJR. Nevertheless, we will use $\hat{\alpha}$ as a metric for two reasons: first, verifying whether $\alpha\geq 1$ (i.e., whether a committee satisfies EJR and OAS) is computationally hard \citep{ABCE+17}, which makes it impractical for evaluation; and the stronger condition $\hat{\alpha}\geq 1$ provides the additional benefit that EJR and OAS can easily be verified through \Cref{lemLowDeltaImpliesStuff}. Second, one could argue that $\hat{\alpha}$ is a meaningful quantity in its own right since it (or rather its inverse $1/\hat{{\alpha}}$) measures how much voter satisfaction could be improved by adding another candidate (giving lower weight to voters who already have many approved candidates).

\paragraph{Polis Results.}
In \Cref{fig:boxplot}, we show the $\hat{\alpha}$ achieved on all the Polis datasets for each of the four algorithms.
\begin{figure}[t]
    \centering
    \includegraphics[width=\linewidth]{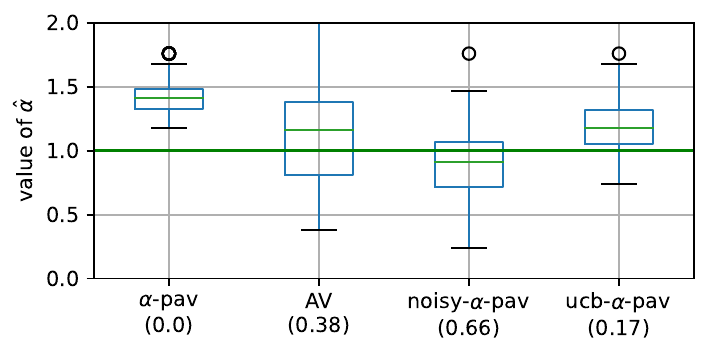}
    \caption{Boxplots where datapoints
        correspond to the $39$ Polis problems ($\times 10$ random seeds). The top / bottom whiskers indicate the maximal / minimal points (except outliers, which are marked by circles), the line in the middle is the median, and the bottom and top of the boxes are the 1st and 3rd quartiles, respectively. The numbers in parenthesis are the fractions of problems where the respective algorithm yields a $\hat{\alpha}\le 1$.
    } \label{fig:boxplot}
\end{figure}
Recall that higher $\hat{\alpha}$ is better and that $\hat{\alpha} \ge 1$ implies OAS and EJR.
As expected, $\alpha$-PAV performs best since it has access to exact queries. Note that it often achieves an $\alpha$ substantially larger than $1$, which means that the corresponding instance allows for better representation than can be guaranteed in the worst case.
AV performs surprisingly well in most experiments, but
in $38\%$ of the cases, it yields $\hat{\alpha}$ smaller than $1$ (and sometimes much smaller). We conclude that for some datasets, it is important to take the complementarity of candidates into account rather than selecting them individually.
The challenge for the proposed algorithms is to do so while being sample-efficient. We see that noisy-$\alpha$-PAV often fails to achieve an $\hat{\alpha}\ge 1$. We know from \Cref{thm:noisy-a-pav-bound} that given enough queries, noisy-$\alpha$-PAV achieves $\hat{\alpha}\ge 1$, so this failure is due to the low number of queried voters. By contrast, ucb-$\alpha$-PAV
yields $\hat{\alpha}\ge 1$ in $83\%$ of the cases, and $\hat{\alpha}\ge 0.75$ in all cases, which indicates that the proposed extensions (i.e., querying promising candidates more often, swapping as soon as possible, and reusing voters) indeed lead to more efficient use of data.

\paragraph{Reddit Results.}
To illustrate why approval voting can perform poorly despite having
access to the full votes, we execute the algorithms on the Reddit dataset described above.
\begin{figure}
    \centering
    \includegraphics[width=\linewidth]{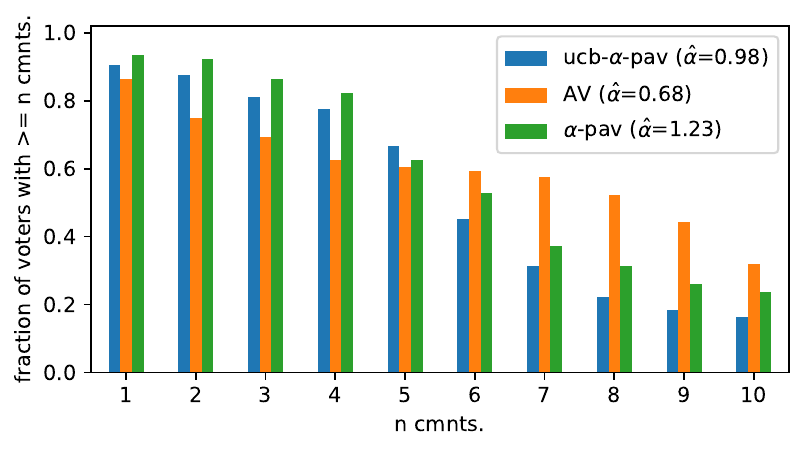}
    \caption{Results on Reddit dataset (with $L=608, m=2135, k=10$): the fraction of voters ($y$-axis) that approve of at least $1,2,..,10$ candidates ($x$-axis) among the selected committee of size $k=10$.}
    \label{fig:reddit}
\end{figure}
In this experiment, AV achieves only $\hat{\alpha}=0.68$. To understand why this happens, we show in \Cref{fig:reddit} the fraction of voters who have at least $1,\ldots, 10$ approved comments in the committee.
We see that AV yields a committee where a high fraction of voters approve many candidates, e.g., about $60\%$ of voters approve $7$ or more candidates, whereas for
$\alpha$-PAV, this is the case for only about $40\%$. This comes at the cost of a high fraction of voters who are poorly represented by AV, e.g., about $25\%$ of voters get at most one approved candidate, whereas for $\alpha$-PAV, this percentage is less than $10\%$. This is to be expected as approval voting does
not take the complementarity of candidates into account and can therefore
lead to less equitable results.
Finally, we observe that ucb-$\alpha$-PAV achieves an $\hat{\alpha}$ close to $1$, and its approval fractions look similar to $\alpha$-PAV, i.e., more equitable than AV. It is interesting that ucb-$\alpha$-PAV performs well on this example, since it only has access to $t=20$ votes for each of the $L=608$ queried candidates, while it has to select from a large number of comments, $m=2135$.

\section{Discussion}\label{secConclusion}

This work bridges the gap between online civic-participation systems, such as Polis, and committee-election methods by enabling them to handle incomplete votes. To deploy the proposed algorithms on such platforms, two practical issues must be considered.

First, our adaptive approach requires control over what the Polis creators call \emph{comment routing}~\cite{SBES+21}: the algorithm that decides which comments are shown to which participants. If on a given platform a comment-routing algorithm is already in place, shared control is possible: each algorithm could determine part of the slate of comments shown to a participant, or the participants themselves can be divided between the algorithms.

Second, in our analysis, we assumed that all comments have been submitted\emdash or all candidates are known\emdash at the time we run our algorithms.
Nevertheless, our algorithms can be extended straightforwardly to a growing set of comments, but we would inevitably lose the representation guarantees for comments that were submitted late if not enough participants could vote on them. In practice, this could be resolved by setting a comment submission deadline, which has been done previously by Polis.

An alternative to our approach would be to complete partial approval votes using collaborative filtering~\cite{RV97}. The completed approval votes can then be aggregated through any approval-based committee election rule, such as PAV. The disadvantage of this approach is that it is unlikely to lead to worst-case guarantees of the type we establish in this paper.

Finally, we emphasize that our approach may be applicable to social media more generally. For instance, as mentioned in \Cref{secEmpirical}, Reddit users also approve or disapprove comments through upvotes and downvotes. However, Reddit uses these inputs to produce a \emph{ranking} of the comments, in contrast to our goal of selecting a subset. There is work on obtaining justified-representation-type guarantees for rankings~\cite{SLBP+17}, which could possibly be extended to the setting of incomplete votes using the techniques developed in this paper. More broadly, this article provides insights into how to fairly represent opinions of groups given incomplete information, which may be relevant for the design of more constructive online ecosystems.

\bibliography{abb,references,ultimate}

\clearpage

\appendix

\if\onlyBody0
    \onecolumn
\begin{center}
    {\LARGE\bf Appendix}
\end{center}
\bigskip

\section{Proof of Theorem~\ref{thmNonadaptiveAlgorithmKEqualsTPlust2}}\label{appT2K3JR}

\NonadaptiveJR*

\begin{proof}
    Consider Algorithm~\ref{algNonadaptiveGreedy}, described below.

    \begin{algorithm}[H]
        \caption{$(k, t)$-non-adaptive (for $t \geq \frac{2}{3}k$)}
        \label{algNonadaptiveGreedy}
        \begin{algorithmic}[1]
            \State Query every set of candidates of size $t$\;
            \For{$i \gets 1, 2, \dots, t$}
            \State
            $c_i \gets$ approval winner among voters not approving $\seq{c}{i - 1}$\;
            \EndFor
            \For{$i \gets t + 1, \ldots, \lfloor \frac{3}{2} t \rfloor$}
            \State
            $c_{i} \gets$ arbitrary default candidate\;
            \For{$c \in C$}
            \State
            $A \gets$ set of voters approving of $c$ but not any of $\{\seq{c}{t - 1}\}$\;
            \State
            $B \gets$ set of voters approving of $c$ but not any of $\{\seq{c}{i-1}\} \setminus \{c_{\lfloor t/2 \rfloor}, \ldots, c_{t-1}\}$\;
            \If{$\abs{A} \geq \frac{n}{k}$ and $\abs{B} \geq \frac{n}{k}$}
            \State
            $c_{i} \gets c$\;
            \EndIf
            \EndFor
            \EndFor
            \State \Return $\{\seq{c}{k}\}$\;
        \end{algorithmic}
    \end{algorithm}

    Provided that $t \geq \frac{2}{3} k$, it is straightforward to verify that the \textbf{if} condition can be checked using only information about sets of voters of size $t$.
    Thus, Algorithm~\ref{algNonadaptiveGreedy} is indeed a non-adaptive $(k, t)$-committee selection algorithm with exact queries.

    For each $i \in \{1, 2, \dots, \lfloor \frac{3}{2} t \rfloor\}$, we say that a voter is \emph{satisfied on round $i$} if it approves of $c_i$, but none of the previously selected candidates $\seq{c}{i - 1}$, and we say that a voter is \emph{satisfied by round $i$} if it was satisfied on some round $j \leq i$. We prove that the final committee satisfies JR by counting the fraction of voters that are satisfied on each round. Indeed, JR is equivalent to the property that there is no 1-cohesive $\frac{1}{k}$-fraction of voters that is left unsatisfied by the $k\tth$ round.

    The case where $k \leq t$ is easy: on each of the first $k$ rounds, either we satisfy a $\frac{1}{k}$ fraction of voters, or there is no 1-cohesive set of $\frac{n}{k}$ unsatisfied voters. Thus, by round $k$, either all voters are satisfied, or the remaining set of unsatisfied voters has no 1-cohesive set of size $\frac{n}{k}$.

    Now suppose that $k > t$. For each $i$, let $x_i$ denote the fraction of voters that are satisfied on round $i$.
    Note that the sequence of $x_i$ are weakly decreasing for $i \leq t$.
    Again, if any $x_i < \frac1k$, it means that there is no 1-cohesive $\frac1k$-fraction of unsatisfied voters after round $i$, so JR is already satisfied. So assume each $x_i \geq \frac1k$ for all $i \leq t$.
    Further, if on any round $i > t$ we fail to find a candidate $c$ making the \textbf{if} condition true, we claim that JR is already satisfied. For if JR were not satisfied, then there would be some candidate $c$ approved by a $\frac1k$-fraction of voters $S$ who approve of no previous candidates. Clearly, we would then have $S \subseteq A$ and $S \subseteq B$, so $A$ and $B$ both contain at least $\frac1k$ fractions of voters.

    Thus, we may assume that, for each $i>t$, candidate $c_i$ satisfies the \textbf{if} statement on round $i$. Consider an arbitrary round $i$. Let $A$ and $B$ denote the respective values of the variables on the iteration of the inner loop where $c_i$ was set to its ultimate value. Observe that the candidates enumerated in the definitions of $A$ and $B$ cover all previously selected candidates. This means that $A \cap B$ is precisely the set of voters approving $c_{i}$ and not any of the previous candidates; in other words, $A \cap B$ is the set of voters satisfied on round $i$. On the other hand, since candidates $\seq{c}{\lfloor t/2 \rfloor}$ are enumerated in the definitions of both sets, it follows that $A \cup B$ is a set of voters approving $c_{i}$ but not any of $\seq{c}{\lfloor t/2 \rfloor}$.
    This means that $A \cup B$ can contain at most an $x_{\lfloor t/2 \rfloor}$ fraction of voters, for otherwise candidate $c_i$ should have been selected earlier, on round $\lfloor t/2 \rfloor$.
    Thus, we may lower bound the fraction of voters satisfied on round $i$ as
    $$\frac{1}{n}\left(\abs{A \cap B}\right) = \frac{1}{n}\left(\abs{A} + \abs{B} - \abs{A \cup B}\right) \geq \frac{1}{n}\left(\frac{n}{k} + \frac{n}{k} - nx_{\lfloor t/2 \rfloor}\right) = \frac{2}{k} - x_{\lfloor t/2 \rfloor}.$$

    Summing over each of the first $k$ rounds, the number of satisfied voters is
    \begin{align*}
        \sum_{i = 1}^{k} (\text{\# satisfied voters on round $i$}) & \geq \sum_{i = 1}^t x_i + \sum_{i = t+1}^{k} \left(\frac{2}{k} - x_{\lfloor t/2 \rfloor}\right)                  \\
                                                                   & = \sum_{i = 1}^{t} x_i - \sum_{i=t+1}^k x_{\lfloor t/2 \rfloor} + (k-t) \frac{2}{k}                              \\
                                                                   & = \sum_{i = 1}^{k - t} \left( x_i - x_{\lfloor t/2 \rfloor}\right) + \sum_{i = k-t+1}^k x_i + (k-t) \frac{2}{k}.
        \intertext{Since the $x_i$ are decreasing and $k-t \leq \lfloor t/2 \rfloor$ when $t \geq 2k/3$, this first term is nonnegative, and $x_i \geq 1/k$ for $i \leq t$ in this case. Therefore we have}
        \sum_{i = 1}^{k} (\text{\# satisfied voters on round $i$}) & \geq \sum_{i = k-t+1}^k x_i + (k-t) \frac{2}{k}                                                                  \\
                                                                   & \geq (t-(k-t)) \frac{1}{k} + (k-t) \frac{2}{k}                                                                   \\
                                                                   & = \frac{2(k - t) + (2t - k)}{k}                                                                                  \\
                                                                   & = 1.
    \end{align*}
    Since all voters are satisfied by round $k$, the final committee satisfies JR.
\end{proof}

\section{Proofs of lower bounds for non-adaptive algorithms with exact queries}\label{appComputationalSearch}

In this section, we prove Theorems~\ref{thmExactLowerBoundN2} and~\ref{thmExactLowerBoundN11}. Both theorems can be derived as applications of the following lemma, which formalizes the properties we require of the adversarial instances shown in Figure~\ref{figMergedCounterexamples}. We only use statement \ref{itmGeneralLPLargeTConclusion} in this paper, but we include statement \ref{itmGeneralLPSmallTConclusion} as well because we believe it may be of independent interest.

\begin{lem}\label{lemLowerBounds}
    Suppose that, for some integers $0 \leq h \leq k_0 \leq \ell$, there exists a probability distribution $\{x_S\}_{S \subseteq [\ell]}$ over subsets of $[\ell]$ such that:
    \begin{enumerate}[label={(\arabic*)}]
        \item\label{itmGeneralLPMarginalConstraint} For any sets $T_1, T_2 \subseteq [\ell]$ such that $\abs{T_1} = \abs{T_2} \leq h$,
              $$\sum_{S \supseteq T_1} x_S = \sum_{S \supseteq T_2} x_S.$$
        \item\label{itmGeneralLPUnhappyConstraint} For some $s^* \in [\ell]$, $x_{\{s^*\}} \geq \frac{1}{k_0}$.
    \end{enumerate}
    Then there exist exact query adversaries for which:
    \begin{enumerate}[label={(\roman*)}]
        \item\label{itmGeneralLPSmallTConclusion} For any $t \leq h$, any non-adaptive $(k, t)$-committee selection algorithm satisfies JR with probability at most $\left(\frac{k_0}{\ell}\right)^{\lfloor k / k_0 \rfloor}$.
        \item\label{itmGeneralLPLargeTConclusion} For any $t > h$, for any $\delta > 0$, any non-adaptive $(k, t)$-committee selection algorithm that makes fewer than $\Omega(m^{h + 1})$ queries satisfies JR with probability at most $\left(\frac{k_0}{\ell}\right)^{\lfloor k / k_0 \rfloor} + \delta$.
    \end{enumerate}
\end{lem}

\begin{proof}
    Given such a probability distribution $\{x_S\}_{S \subseteq [\ell]}$, we define the query adversary as follows.
    This adversary will be a distribution over profiles over $[m]$ candidates, for some sufficiently large $m$ to be determined later.
    We denote a given non-adaptive $(k, t)$-committee selection algorithm by $\mathcal{A}$.

    First let $k = p k_0 + r$, where $p$ is a nonnegative integer and $0 \leq r < k_0$.
    Partition the candidates into $[m] = C_1 \cup C_2 \cup \dots \cup C_p \cup D$ where, for each $i \in [p]$, $\abs{C_i} = \lfloor (m - r)/p \rfloor$, and $D$ contains the remaining candidates, of which there are at least $r$.
    For each $i$, the adversary will randomly select a subset of $\ell$ distinct candidates $S^i := \{c^i_1, c^i_2, \dots, c^i_{\ell}\} \subseteq C_i$.
    The adversary chooses all subsets and orderings with equal probability, independently for each $C_i$.
    The adversary will then respond to all queries according to the following approval matrix.

    We partition the voters into $p + r$ distinct ``parties'' $P_1, \ldots, P_p$ and $Q_1, \ldots, Q_r$, and every voter is a member of exactly one party.
    Each party $P_i$ contains a $k_0/k$ proportion of voters, and voters in $P_i$ approve only of some subset of the candidates contained in $S^i \subseteq C_i$, and none of the other candidates.
    For these $P_i$, for all $S \subseteq [\ell]$, let the fraction of voters whose approval set is exactly $\{c^i_s \suchthat s \in S\}$ be equal to $x_S$.
    Every party $Q_j$ is a $1/k$ proportion of the voters, and each voter belonging to $Q_j$ approves only of one candidate $d_j \in D$ which is specific to $Q_j$.


    Let us say that $\mathcal{A}$ \emph{$h$-covers} a given set of candidates $S\subset [m]$ if $\mathcal{A}$ ever submits a query $T\subseteq [m]$ such that $\abs{T \cap S} > h$.
    If, for any of the parties $P_i$ with $i \in [p]$, the algorithm fails to $h$-cover the set $S^i$, then condition~\ref{itmGeneralLPMarginalConstraint} implies that all $\ell$ of these candidates are completely symmetric (i.e. indistinguishable) to $\mathcal{A}$ given all of its query responses.
    Since each of the distinguished candidates $c^i_{s^*}$ is distributed uniformly at random among the candidates $S^i$, $\mathcal{A}$ selects $c^i_{s^*}$ to be part of its chosen committee with probability at most $\min(k_i/\ell, 1)$, where $k_i$ is the number of candidates that $\mathcal{A}$ selects from $P_i$.

    However, in order to satisfy JR, $\mathcal{A}$ must select at least $r$ candidates from $D$, since there are $r$ distinct candidates in $D$ approved by the $r$ parties $Q_j$, which are disjoint fractions of $1/k$ of the voters.
    In order to satisfy JR $\mathcal{A}$ must also select the distinguished candidate $c^i_{s^*} \in C_i$ for each party $P_i$, since condition~\ref{itmGeneralLPUnhappyConstraint} implies that for each $P_i$ at least a $\frac{1}{k_0} \cdot \frac{k_0}{k} = \frac{1}{k}$ fraction of the voters approve only $c^i_{s^*}$ and none of the other candidates.

    This already implies \ref{itmGeneralLPSmallTConclusion}: assuming that $\mathcal{A}$ selects at least $r$ candidates from $D$, then if $t \leq h$, it is impossible for $\mathcal{A}$ to $h$-cover $S^i$ with any number of queries, and thus $\mathcal{A}$ succeeds in satisfying JR with probability at most
    $$
        \Pr[\mathcal{A} \text{ selects } c_{s^*}^i \text{, for all } i \in [p]] \leq \frac{k_1}{\ell} \cdot \frac{k_2}{\ell} \cdot \dots \cdot \frac{k_p}{\ell} = \frac{k_1 k_2 \dots k_p}{\ell^p} \leq \frac{k_0^p}{\ell^p} = \left(\frac{k_0}{\ell}\right)^p.
    $$
    Here the second inequality holds due to the constraint that $k_1 + k_2 + \dots + k_p \leq k - r = p k_0$, since $\mathcal{A}$ must select at least $r$ candidates from $D$.

    \medskip
    To prove \ref{itmGeneralLPLargeTConclusion}, we must analyze the likelihood that an algorithm $\mathcal{A}$ making a small number of queries $h$-covers any given $S^i$.
    Let us suppose that $\mathcal{A}$ knows the partition of candidates into $C_1 \cup C_2 \cup \dots \cup C_p \cup D$, knows everything about the approval matrix except for which sets $S^i$ were chosen within each party $P_i$, and is allowed to make at most $cm^{h + 1}$ queries within each party $P_i$, separately, where
    $$
        c := \frac{\delta}{2^{t + \ell} \: \ell! \: p^{h + 2}}.
    $$
    Clearly, these assumptions only make the algorithm $\mathcal{A}$ stronger: an impossibility for this kind of algorithm implies the desired lower bound.
    Fix a party $P_i$.
    For sufficiently large $m$, every set $S\subseteq C$ of $\ell$ candidates that is $h$-covered by a query $T\subseteq [m]$ of size $t$ can be decomposed into two parts: a set of size $j$ (where $h + 1 \leq j \leq t$) that is contained in $T$, and a set of size $\ell - j$ that is contained in $C_i \setminus T$.
    Thus, the number of sets $S$ of size $\ell$ within $C_i$ that any single query can $h$-cover is exactly
    $$
        \sum_{j = h + 1}^{t} \binom{t}{j} \binom{\lfloor (m - r) / p \rfloor - t}{\ell - j} \leq 2^t \binom{m/p}{\ell - h - 1} \leq 2^t \left(\frac{m}{p}\right)^{\ell - (h + 1)}
    $$
    provided that $\ell - h - 1 \leq m/2p$, which holds for sufficiently large $m$.

    Since we have that $\mathcal{A}$ made at most $cm^{h + 1}$ queries within party $P_i$ by assumption, at most
    $$
        cm^{h + 1} \cdot 2^t \left(\frac{m}{p}\right)^{\ell - (h + 1)} = \frac{2^t}{p^{\ell - (h + 1)}} cm^\ell
    $$
    sets of size $\ell$ can be $h$-covered.
    Since, within each party $P_i$ there are a total of
    $$
        \binom{\lfloor (m - r) / p \rfloor}{\ell} \geq \frac{(\lfloor (m - r) / p \rfloor - \ell)^\ell}{\ell!} \geq \frac{(m/(2p))^\ell}{\ell!}
    $$
    (for sufficiently large $m$) sets of size $\ell$ in total, and each one of them is chosen to be $S_i$ by the adversary with equal probability, the likelihood that $\mathcal{A}$ $h$-covered $S_i$ is at most
    $$\frac{2^t cm^\ell \ell!}{p^{\ell - (h + 1)}(m/(2p))^\ell} = 2^{t + \ell} \ell! p^{h + 1} c = \frac{\delta}{p}.$$
    It follows from a union bound over all $p$ of the parties $P_i$ that the probability that $\mathcal{A}$ $h$-covered \emph{any} of the $S_i$ is at most $\delta$.
    For $\mathcal{A}$ to satisfy JR, it is necessary for it to either $h$-cover some $S_i$ with the initial queries or subsequently select every $c^i_{s^*}$ after having failed to $h$-cover any $S_i$.
    By the union bound, the probability that $\mathcal{A}$ satisfies JR is at most the sum of the probabilities of these two events, which is at most $\left(\frac{k_0}{\ell}\right)^p + \delta$.
\end{proof}

Thus, to prove lower bounds against non-adaptive algorithms with exact queries, it suffices to construct probability distributions over subsets of a finite set $[\ell]$ with certain special properties. To prove our $\Omega(m^2)$ lower bound, which holds for any $k \geq 2$, we generalize the construction from Figure~\ref{figMergedCounterexamples} (a) by simply adding more candidates to the larger approval set:

\ExactLowerSquare*

\begin{proof}
    Given any $t$ and $\varepsilon \in (0, 1]$, let $h = 1$, $k_0 = 2$, and $\ell = \lceil 4/\varepsilon \rceil$. Consider the probability distribution over subsets of $[\ell]$ where $x_{\{1\}} = \frac12$, $x_{\{2, 3, 4, \dots, \ell\}} = \frac12$, and all other sets have probability zero (Figure~\ref{figMergedCounterexamples} (a) shows the special case of this distribution where $\ell = 3$).
    Notice that these parameters meet all the requirements of Lemma~\ref{lemLowerBounds} with $s^* = 1$. Letting $\delta := \varepsilon/2$, it follows that, for any $k \geq k_0 = 2$, any $(k, t)$-committee selection algorithm that makes fewer than $\Omega(m^2)$ queries satisfies JR with probability at most
    \begin{equation*}
        \left(\frac{k_0}{\ell}\right)^{\lfloor k / k_0 \rfloor} + \delta \leq \left(\frac{k_0}{\ell}\right)^{1} + \delta = \frac{2}{\ell} + \delta \leq \frac{\varepsilon}{2} + \frac{\varepsilon}{2} = \varepsilon. \qedhere
    \end{equation*}
\end{proof}

To prove stronger lower bounds we need to increase the $h$ parameter. Probability distributions $\{x_S\}_{S \subseteq [\ell]}$ satisfying the hypotheses of Lemma~\ref{lemLowerBounds} prove difficult to construct by hand for $h > 1$, so we conducted a computational search. By a straightforward averaging argument, one can see that it is without loss of generality to consider ``symmetric'' distributions, where for any sets $S, T \subseteq [\ell]$ of the same size that either both contain $s^*$ or both do not contain $s^*$, $x_{S} = x_{T}$. Thus, it suffices to consider solutions encoded as points in the following polyhedron, which we refer to as $P(h, k_0, \ell) \subseteq \rr^{2\ell}$. We parameterize the space by the $2\ell$ variables
$$\{x_{i, j} \suchthat i \in \{0, 1\}, j \in \{0, 1, 2, \dots, \ell - 1\}\},$$
where $x_{0, j}$ encodes the value of $x_S$ for all $S$ of size $j$ that do not include $s^*$, and $x_{1, j}$ encodes the value of $x_S$ for all $S$ containing $s^*$ and $j$ other elements from $[\ell]$. For a solution to be in $P(h, k_0, \ell)$, there are four kinds of constraints it must satisfy.
\begin{itemize}
    \item All probabilities must be nonnegative: for all $i \in \{0, 1\}$ and $j \in \{0, 1, 2, \dots, \ell - 1\}$,
          \begin{equation*}
              x_{i, j} \geq 0.
          \end{equation*}
    \item Probabilities must all sum to 1:
          \begin{equation*}
              \sum_{i = 0}^1 \sum_{j = 0}^{\ell - 1} \binom{\ell - 1}{j} x_{i, j} = 1.
          \end{equation*}
    \item Condition~\ref{itmGeneralLPMarginalConstraint} from Lemma~\ref{lemLowerBounds} must be satisfied. Due to the symmetry that is baked in to the solutions we're considering, we only need to check the constraint for pairs of sets where $s^*$ is contained in one set but not the other. This constraint is as follows: for all $t' \in [h]$,
          \begin{equation*}
              \sum_{i = 0}^1 \sum_{j = t'}^{\ell - 1} \binom{\ell - 1 - t'}{j - t'} x_{i, j} = \sum_{j = t' - 1}^{\ell - 1} \binom{\ell - t'}{j - t' + 1} x_{1, j}.
          \end{equation*}
    \item Condition~\ref{itmGeneralLPUnhappyConstraint} from Lemma~\ref{lemLowerBounds} must be satisfied:
          \begin{equation*}
              x_{1, 0} \geq \frac{1}{k_0}.
          \end{equation*}
\end{itemize}

Thus, there exists a probability distribution satisfying the hypotheses of Lemma~\ref{lemLowerBounds} if and only if $P(h, k_0, \ell)$ is nonempty. The description of this polyhedron is only of polynomial-size, so we can solve it efficiently using linear programming. However, many of the coefficients are extremely large, and we eventually ran into numerical difficulties. Table~\ref{tabLowerBoundsComputationalSearch} lists the tightest lower bounds we were able to obtain in terms of how large $k_0$ had to be for a given value of $h$, and Table~\ref{tabM11Solution} provides one point in $P(10, 72, 73)$ as an example.

\begin{table}[H]\centering\begin{tabular}{r|c c c c c c c c c c c}
        $h$   & 1 & 2 & 3  & 4  & 5  & 6  & 7  & 8  & 9  & 10 \\\hline
        $k_0$ & 2 & 6 & 10 & 16 & 21 & 30 & 38 & 49 & 59 & 72
    \end{tabular}\caption{\label{tabLowerBoundsComputationalSearch}For each positive integer $h$, smallest value of $k_0$ for which $P(h, k_0, k_0 + 1)$ is nonempty, i.e., the smallest committee size for which Lemma~\ref{lemLowerBounds} implies that guaranteeing JR requires $\Omega(m^{h + 1})$ exact queries of size $t > h$. The constraints are tight only for $h \in \{1, 2\}$.}\end{table}

\newcommand{\te}[2]{{#1}e-{#2}}
\bgroup
\setlength\tabcolsep{0.37em}
{\tiny\begin{table}[H]\centering\begin{tabular}{c c c c c c c c c c c}
            $x_{1, 0}$ & $x_{0, 2}$    & $x_{1, 5}$    & $x_{0, 13}$    & $x_{1, 20}$    & $x_{0, 31}$    & $x_{1, 41}$    & $x_{0, 52}$    & $x_{1, 59}$    & $x_{0, 67}$    & $x_{1, 70}$   \\\hline
            0.01398    & \te{3.204}{4} & \te{1.184}{9} & \te{1.012}{15} & \te{4.781}{20} & \te{7.926}{23} & \te{5.799}{23} & \te{1.875}{20} & \te{2.058}{16} & \te{8.968}{11} & \te{4.577}{6}
        \end{tabular}\caption{\label{tabM11Solution}The point in $P(10, 72, 73)$ maximizing $x_{1, 0}$. All variables not shown in the table have value zero.}\end{table}}
\egroup

\ExactLowerNEleven*

\begin{proof}
    Let $\varepsilon > 0$ be given. Then let
    $$k := 72\left(\frac{\log(2/\varepsilon)}{\log(73/72)} + 1\right)$$
    and $\delta := \varepsilon/2$. As Tables~\ref{tabLowerBoundsComputationalSearch} and \ref{tabM11Solution} show, Lemma~\ref{lemLowerBounds} holds for $h = 10$, $k_0 = 72$, and $\ell = 73$. Thus, for any $t$, any non-adaptive $(k, t)$-committee selection algorithm that makes fewer than $\Omega(m^{11})$ queries satisfies JR with probability at most
    \begin{equation*}
        \left(\frac{72}{73}\right)^{\lfloor k / 72 \rfloor} + \delta \leq \left(\frac{72}{73}\right)^{(k / 72) - 1} + \delta = \frac{\varepsilon}{2} + \frac{\varepsilon}{2} = \varepsilon. \qedhere
    \end{equation*}
\end{proof}

We note that there is a gap between these results and Theorem \ref{thmExactQueryUpper}. An intriguing direction for future work is to obtain matching upper and lower bounds for the query complexity of guaranteeing JR using a non-adaptive algorithm with exact queries. For $k = 1$ we need $\Theta(m)$ queries, and for $k \in \{2, 3\}$, we need $\Theta(m^2)$ queries. However, the complexity is unknown for all larger $k$, and we conjecture that the exponent of $m$ grows as a polynomial function of $k$.

\section{Family of examples for noisy vs exact queries}\label{appNoisyVsExact}

Fix some $k \ge 4$, $t$, and $m$. We will construct a family of instances on $m$ candidates where there exists a non-adaptive exact-query algorithm which can guarantee JR using $\ceil*{m/t}$ queries while a non-adaptive noisy-query algorithm necessarily needs $\Omega(m\log(m)/t)$ to guarantee it with any fixed probability $\delta$. We describe the approval profile by the distribution over approval sets by sampling a voter uniformly at random. There is one special candidate $a^*$. This candidate $a^*$ is approved by a $2/k$ fraction of the electorate while all other candidates $b$ are approved by $1/(2k)$. Further, these approvals are independent in the sense that when we sample a voter, the joint distribution over approvals is as if each of these approvals were selected independently. For example, given a set $S \subseteq C$ of candidates such that $S$ contains $a^*$ and $\ell$ other candidates, the proportion of voters who approve exactly the set $S$ is $(2/k) \cdot (1/(2k))^\ell \cdot (1 - 1/(2k))^{m - 1 - \ell}$.

The first observation we make is that the committees that satisfy JR are exactly those that include $a^*$. Notice that if a committee includes $a^*$, no other candidate has enough approval support ($1/k$) to have a blocking coalition to violate JR. On the other hand, if there is a committee $W$ of size $k$ that does not include $a^*$, we can compute the proportion of voters that approve $a^*$ that do not approve of any candidates in $W$. This is
\[
    \frac{2}{k} \cdot \left(1 - \frac{1}{2k} \right)^k > \frac{2}{k} \left(1 - \frac{k}{2k} \right) = \frac{1}{k}.
\]
Hence, for such a $W$, there would exist a sufficiently large blocking coalition for $a^*$.

Next, we show that there is a non-adaptive exact-query algorithm that can guarantee JR for any instance of this form (i.e., regardless of which candidate is $a^*$). Indeed, it simply makes $\ceil*{m/t}$ queries that cover all candidates. From this, it can deduce candidate approval scores and ensure that the committee it chooses contains the candidate with approval score $2/k$.

Finally, let us consider a non-adaptive algorithm that makes $\ell$ queries and guarantees JR with probability $1 - \delta$ regardless of which candidate is $a^*$. Notice that such an algorithm should guarantee JR with this same probability against a distribution of instances where $a^*$ is selected uniformly at random. Let us consider an algorithm $A$ that maximizes the probability of selecting a JR committee against this distribution. Notice that it is without loss of generality that $A$ is deterministic by Yao's minimax principle. We show that this can only be done if $\ell \ge f(m)$ where $f(m) \in \Omega(m \log m)$ (treating $k$, $t$, and $\delta$ as constants) is a function to be defined later.

Suppose for a contradiction $\ell < f(m)$. Let $H$ be the set of candidates that appear in strictly more than $q := \frac{2tf(m)}{m\delta}$ and let $L$ be the remaining candidates. Notice that $|H| \le \delta / 2 \cdot m$, as otherwise $\ell \ge f(m)$. We show that conditioned on $a^* \in L$, the probability $A$ chooses a committee containing $a^*$ is at most $\delta / 2$.  This implies that $A$'s probability of success is at most $(1 - \delta / 2) \cdot \delta / 2 + \delta / 2 < \delta$.

To that end, consider an algorithm that receives extra queries such that all candidates in $L$ are in exactly $q$ queries. Notice that conditioned on $a^*$ being in $L$, since all candidates in $L$ are in the same number of queries, the optimal strategy to maximize the probability $a^*$ is a committee-member is to take the $k$ candidates in $L$ with highest empirical approval score. Indeed, this dominates any other strategy as conditioned on any empirical approval scores, this choice of committee covers the maximum likelihood estimates of the underlying distribution.

What we finally show is that with probability at least $1 - \delta / 2$, conditioned on $a^* \in L$, $a^*$ will \emph{not} be among the $k$ highest approval scores. Intuitively, with reasonably high probability $a^*$ will have empirical not too much more than it's true approval, say at most $3/k$, while, by choosing $q \in O(\log m)$, due to the noise in estimating empirical approval scores, at least $k$ of the remaining candidates in $L$ will have approval score this large. Indeed, ensuring $q > \frac{k^2\log(4/\delta)}{2}$ ensures the empirical estimate of $a^*$ is less than $3/k$ with probability at least $1 - \delta / 4$ using standard Hoeffding's inequality. For the other emprical means, using tail bounds on the Binomial distribution~\citep{informationtheory}, the probability they are at least $3/k$ is at least $\frac{1}{\sqrt{2q}} \exp(-q \Theta(k))$. Notice we can choose $q \in \Theta(\log m)$ such that this value is at least $2(k + \log(4/\delta))/m$. For sufficiently large $m$, this choice of $q$ will be above $\frac{k^2\log(4/\delta)}{2}$, leading to a valid $\Theta(m \log m)$ function of $f$. Further, again applying standard Hoeffding's inequality on the remaining at least $m/2$ candidates in $L$ shows that at least $k$ will satisfy this, as needed.

\section{Proofs for the Adaptive Exact-Query Setting\label{sec:proofs}}

In the following, we prove  \Cref{lemLowDeltaImpliesStuff}, \Cref{lemPotentialIncrease}, and \Cref{thmExactQueryUpper}.

\subsection{Proof of \Cref{lemLowDeltaImpliesStuff}}
\Cref{lemLowDeltaImpliesStuff} and its proof are based on the lower bound from Lemma 1 in \cite{Skowron2021-ys}. Our result is more general in two ways: (1) our statement holds for any committee $W$, no matter what algorithm computed it, and (2) we introduce an approximation parameter $\alpha$.
We begin with the following intermediate lemma:
\begin{lem}
    \label{lem:low_delta_is_good}For any committee $W \subseteq C$ and group of voters
    $V\subseteq N$, we have
    \[
        \text{avs}_{W}(V)\ge\min\left\{ \left|\bigcap_{i\in V}A_{i}\right|,\frac{1}{n} \cdot \frac{|V|}{\pavd^{*}(W)}-1\right\} .
    \]
\end{lem}

\begin{proof}
    As mentioned, the following proof is closely related to the proof of Lemma 1 of \citet{Skowron2021-ys}.
    Suppose there exist $V$ and $W$ such that both
    \[
        \frac{1}{|V|}\sum_{i\in V}|W\cap A_{i}|<\left|\bigcap_{i\in V}A_{i}\right|\text{ and }\frac{1}{|V|}\sum_{i\in V}|W\cap A_{i}|<\frac{1}{n} \cdot \frac{|V|}{\pavd^{*}(W)}-1.
    \]
    We then have
    \begin{align*}
        \left|\bigcap_{i\in V}A_{i}\right| & >\frac{1}{|V|}\sum_{i\in V}|W\cap A_{i}|                                          \\
                                           & \ge\frac{1}{|V|}\sum_{i\in V}\left|W\cap\left(\bigcap_{j\in V}A_{j}\right)\right| \\
                                           & =\left|W\cap\left(\bigcap_{j\in V}A_{j}\right)\right|.
    \end{align*}
    This implies $W\cap\left(\bigcap_{i\in V}A_{i}\right) \subsetneq \bigcap_{i\in V}A_{i}$, so
    \(
    \overline{W} \cap \left(\bigcap_{i\in V}A_{i} \right)  \ne \emptyset\). Hence, there is a candidate $c \in \overline{W} \cap \left(\bigcap_{i\in V}A_{i} \right)$ that is not on the committee $W$, but is approved by all voters in $V$. For such a
    candidate $c$, we have
    \begin{align*}
        \pavd(W,c) & =\frac{1}{n}\sum_{i\in N:c\in A_{i}\setminus W}\frac{1}{|A_{i}\cap W|+1}                                   \\
                   & =\frac{1}{n}\sum_{i\in N:c\in A_{i}}\frac{1}{|A_{i}\cap W|+1}            \tag{$c \notin W$}                \\
                   & \ge\frac{1}{n}\sum_{i\in V}\frac{1}{|A_{i}\cap W|+1} \tag{$c\in\bigcap_{i\in V}A_{i}$}                     \\
                   & \ge\frac{1}{n}|V|\frac{1}{\frac{1}{|V|}\sum_{i\in V}(|W\cap A_{i}|+1)} \tag{convexity of \ensuremath{1/x}} \\
                   & >\frac{1}{n}|V|\frac{1}{\frac{1}{n}\frac{|V|}{\pavd^{*}(W)}-1+1}                                           \\
                   & =\pavd^{*}(W),
    \end{align*}
    a contradiction, as $\Delta(W, x) \le \Delta^*(W)$ for all candidates $x$.
\end{proof}
We are now ready to prove \Cref{lemLowDeltaImpliesStuff}, which we
restate here:
\lemLowDeltaImpliesStuff*
\begin{proof}
    Fix a committee $W$ satisfying $\Delta^*(W) < \frac{1}{\alpha k}$. We begin with $\alpha$-OAS. Fix a $\lambda \in [0, k]$, and a $\frac{\lambda + 1}{\alpha}$-large, $\lambda$-cohesive group of voters $V$. By definition of $\lambda$-cohesive, $\left|\bigcap_{i \in V} A_i \right| \ge \lambda$. Further, we have
    \begin{align*}
        \frac{1}{n} \cdot \frac{|V|}{\Delta^*(W)} - 1
         & \ge \frac{1}{n} \cdot \frac{1}{\Delta^*(W)} \cdot \frac{\lambda + 1}{\alpha}\cdot \frac{n}{k} - 1 \tag{$V$ is $\frac{\lambda + 1}{\alpha}$-large} \\
         & = \frac{1}{\Delta^*(W)} \cdot \frac{\lambda + 1}{\alpha}\cdot \frac{1}{k} - 1                                                                     \\
         & > \lambda + 1 - 1 \tag{$\Delta^*(W) < \frac{1}{\alpha k}$}                                                                                        \\
         & = \lambda
    \end{align*}
    Together, these imply that
    \[
        \min\left\{ \left|\bigcap_{i\in V}A_{i}\right|,\frac{1}{n} \cdot \frac{|V|}{\pavd^{*}(W)}-1\right\} \ge \lambda.
    \]
    Invoking \Cref{lem:low_delta_is_good}, we have $\text{avs}_{W}(V) \ge\lambda$, as needed.

    Next, we show $\alpha$-EJR. Fix an $\ell \in [k]$, and an $\frac{\ell}{\alpha}$-large, $\ell$-cohesive group of voters $V$. As before, by the definition of $\ell$-cohesive, we have $\left|\bigcap_{i\in V}A_{i}\right| \ge \ell$. Further, by the same argument as above with $\ell = \lambda + 1$,
    \[
        \frac{1}{n} \cdot \frac{|V|}{\Delta^*(W)} - 1 > \ell - 1.
    \]
    Together, these imply that
    \[
        \min\left\{ \left|\bigcap_{i\in V}A_{i}\right|,\frac{1}{n} \cdot \frac{|V|}{\pavd^{*}(W)}-1\right\} > \ell - 1.
    \]
    Invoking \Cref{lem:low_delta_is_good}, we have $\text{avs}_{W}(V) > \ell - 1$, and since utilities are integers, this implies that $|A_i \cap W| \ge \ceil*{\text{avs}_{W}(V)} \ge  \ell$ for at least one voter $i \in V$, as needed.
\end{proof}

\subsection{Proof of \Cref{lemPotentialIncrease}}
\Cref{lemPotentialIncrease} here:

\lemPotentialIncrease*
\begin{proof}
    Fix $W$ and $c \notin W$. We will use the notation $W^+ := W \cup \set{c}$.
    First, we show that
    \begin{equation}
        \label{ineq:bound-delta}
        \min_{x \in W} \Delta(W^+ \setminus \set{x}, x) \le \frac{1 - \Delta(W, c)}{k}.
    \end{equation}
    To that end, let us consider $\Delta(W^+ \setminus \set{x}, x)$ for an arbitrary $x \in W^+$. We have
    \begin{align*}
        \Delta(W^+ \setminus \set{x}, x)
         & = \pavsc(W^+) - \pavsc(W^+ \setminus \set{x})                     \\
         & = \frac{1}{n} \sum_{i \in N: x \in A_i} \frac{1}{|W^+ \cap A_i|}.
    \end{align*}
    Adding up over all $x \in W^+$, we have
    \begin{align*}
        \sum_{x \in W^+} \Delta(W^+ \setminus \set{x}, x)
         & = \frac{1}{n} \sum_{x \in W^+} \sum_{i \in N: x \in A_i} \frac{1}{|W^+ \cap A_i|}          \\
         & =  \frac{1}{n} \sum_{i : W^+ \cap A_i \ne \emptyset} \frac{|W^+ \cap A_i|}{|W^+ \cap A_i|} \\
         & \le 1.
    \end{align*}
    On the other hand, we have
    \begin{align*}
        \sum_{x \in W^+} \Delta(W^+ \setminus \set{x}, x)
         & = \Delta(W, c) + \sum_{x \in W} \Delta(W^+ \setminus \set{x}, x)             \\
         & \ge \Delta(W, c) + k \cdot  \min_{x \in W} \Delta(W^+ \setminus \set{x}, x).
    \end{align*}
    Combining these two inequalities, we get that
    \[\Delta(W, c) + k \cdot  \min_{x \in W} \Delta(W^+ \setminus \set{x}, x) \le 1.\]
    Rearranging yields \eqref{ineq:bound-delta}. Finally, notice that
    \[
        \max_{x \in W} \Delta(W, c, x) = \Delta(W, c) - \min_{x \in W} \Delta(W^+ \setminus \set{x}, x).
    \]
    Hence, by \eqref{ineq:bound-delta},
    \[
        \max_{x \in W} \Delta(W, c, x)
        \ge \Delta(W, c) - \left(\frac{1 - \Delta(W, c)}{k} \right)
        = \frac{(k + 1) \Delta(W, c) - 1}{k},
    \]
    as needed.
\end{proof}

\subsection{Proof of \Cref{thmExactQueryUpper}}

\GlsPavGuarantees*
\begin{proof}
    Clearly, if \Cref{alg:a-pav} terminates, the resulting committee $W$ satisfies
    \[
        \pavd^{*}(W)=\pavd(W,c')<\frac{1}{\alpha k}
    \]
    and hence $\alpha$-EJR and $\alpha$-OAS by \Cref{lemLowDeltaImpliesStuff}.
    What remains is to bound how many steps the algorithm takes to terminate. To do this, we use the PAV score of the current committee as a potential function. In every iteration of the loop for which the algorithm does not terminate, we have
    \[
        \pavd(W,c')\ge\frac{1}{\alpha k}
    \]
    and hence, by \Cref{lemPotentialIncrease}, the increase in PAV
    score at each step will be
    \[
        \max_{x}\pavd(W,c',x)\ge\frac{(1-\alpha)k+1}{\alpha k^{2}}.
    \]
    Notice that the minimum and maximum PAV score that can be possibly attained
    by any committee of size $k$ are $0$ (when nobody approves of any
    candidate) and $H_{k}$ (the harmonic number which is attained when everyone approves of every candidate)
    respectively. Hence, there can be at most
    \[
        \frac{\alpha k^{2}}{(1-\alpha)k+1}H_{k}
    \]
    steps. Since at each step we make
    \(
    \ceil*{\frac{m-k}{t-k}}
    \)
    queries, the result follows.
\end{proof}

\section{PAV and LS-PAV Yield a Committee That Satisfies $\pavd^{*}(W)<1/k$\label{sec:pav_proofs}}

As mentioned previously, it is known that PAV and LS-PAV satisfy both
EJR and OAS. Here, we show that they yield committees that satisfy
$\pavd^{*}(W)<1/k$, which implies EJR and OAS through
\Cref{lemLowDeltaImpliesStuff}. This is noteworthy because 1) it
has a much simpler proof, 2) it implies that PAV and LS-PAV committees
can be certified in a computationally efficient manner, by verifying
that $\pavd^{*}(W)<1/k$.
\begin{lem}
    For both PAV and LS-PAV, the returned committee $W$
    satisfies $\pavd^{*}(W)<1/k$.
\end{lem}

\begin{proof}
    For the committee $W$ computed by PAV or LS-PAV, we have that for any
    candidate $c\notin W$,
    \[
        \max_{x\in W}\pavd(W,c,x) < \frac{1}{k^{2}}
    \]
    since otherwise the PAV score of $W$ could be improved by at least
    $1/k^{2}$ by adding $c$ and removing the worst candidate. By \Cref{lemPotentialIncrease}, this implies that
    \begin{align*}
        \frac{(k+1)\pavd^{*}(W)-1}{k} & <\frac{1}{k^{2}}.
    \end{align*}
    Rearranging yields
    $\pavd^{*}(W)<\frac{1}{k}$, so the result follows from \Cref{lemLowDeltaImpliesStuff}.
\end{proof}

\section{Proof of \Cref{thm:noisy-a-pav-bound}}\label{sec:noisy-a-pav-bound}
\noisyAPAVBound*

\begin{proof}
    We first show that with probability at least $1 - \delta$, all $\hat{\Delta}$ estimates in the first $H_k \cdot \frac{2\alpha k^2}{(1 - \alpha)k + 1}$ iterations of the loop are within $\pm \varepsilon := \frac{(1 - \alpha)k + 1}{12 \alpha k^2}$ of the corresponding true $\Delta$ values. Then, we show that conditioned on these accurate estimates, the algorithm satisfies the theorem properties.

    To show the error bounds, we will use a straightforward application of Hoeffding's inequality. Indeed, when the corresponding $\ell$ voters are sampled, notice that $\hat{\Delta}$ is simply the sample mean of independent samples with expectation of the corresponding $\Delta$. Further, these samples are always proportions falling in $[-1, 1]$. Hence, any specific estimate will not be within $\pm \varepsilon$ with probability at most $2\exp(-\varepsilon^2 \ell/2)$. Note that there are $m - k$ choices of $x$ for $\hat{\Delta}(W, x)$ and $(m - k) \cdot k$ choices of $(x, y)$ pairs for $\hat{\Delta}(W, x, y)$. Hence, there are a total of $(m - k)(k + 1)$ estimates per iteration. Therefore, there are at most $H_k \cdot \frac{2\alpha k^2}{(1 - \alpha)k + 1}(m - k)(k + 1)$ estimates in the first $H_k \cdot \frac{2\alpha k^2}{(1 - \alpha)k + 1}$ iterations. A union bound tells us the probability that all estimates in these iterations are within $\pm \varepsilon$ is at least \[1 - 2\exp(-\varepsilon^2 \ell/2) \cdot H_k \cdot \frac{2\alpha k^2}{(1 - \alpha)k + 1}(m - k)(k + 1).\]
    We simply need to show that this value is at least $1 - \delta$.

    To that end, recall that
    \[\ell = \ceil*{288  \left( \frac{\alpha k ^2}{(1 - \alpha)k + 1} \right)^2 \log \left(\frac{8mk^4}{\delta} \right)}.\]
    Noting that $H_k \le k $ and $\frac{\alpha  k^2}{(1 - \alpha)k + 1} \le k^2$, we have that
    \[8mk^4 \ge 2 k \cdot (2k^2) \cdot m \cdot (2k) \ge  2 H_k \cdot \frac{2\alpha k^2}{(1 - \alpha)k + 1}(m - k)(k + 1).\]
    Hence,
    \[\ell \ge \frac{2}{\varepsilon^2} \log \left(\frac{2H_k \cdot \frac{2\alpha k^2}{(1 - \alpha)k + 1}(m - k)(k + 1)}{\delta} \right),\]
    so we have
    \[
        2\exp\left(-\varepsilon^2 \ell/2 \right) \le \frac{\delta}{H_k \cdot \frac{2\alpha k^2}{(1 - \alpha)k + 1}(m - k)(k + 1)},
    \]
    as needed.

    Next, condition on all of these estimates being accurate. Notice if the algorithm terminates within the first $H_k \cdot \frac{2\alpha k^2}{(1 - \alpha)k + 1}$ iterations, this means that for the returned committee, $\max_{x \notin W} \hat{\Delta}(W, x) < \frac{1}{\alpha k} - \frac{(1 - \alpha)k + 1}{8\alpha k^2} = \frac{1}{\alpha k} - \varepsilon$. By our assumption about the accuracy of each $\hat{\Delta}$, we have that $\Delta^*(W) = \max_{x \notin W} \hat{\Delta}(W, x) < \frac{1}{\alpha k}$. Hence, by \Cref{lemLowDeltaImpliesStuff}, $W$ satisfies the desired properties. Further, there are $\ceil*{\frac{m - k}{t -k}} \cdot \ell$ queries per iteration. Noting that
    \[
        \ell \le 289  \left( \frac{\alpha k ^2}{(1 - \alpha)k + 1} \right)^2 \log \left(\frac{8mk^4}{\delta} \right)
    \] to avoid the ceiling, this means the total query complexity is at most
    \[
        H_k \cdot \frac{2\alpha k^2}{(1 - \alpha)k + 1} \cdot \ceil*{\frac{m - k}{t -k}} \cdot \ell \le 578 H_k \ceil*{\frac{m - k}{t -k}}
        \left( \frac{\alpha k ^2}{(1 - \alpha)k + 1} \right)^3 \log \left(\frac{8mk^4}{\delta} \right)\]
    as needed.

    What remains to be shown is that conditioned on the accurate estimates, the algorithm terminates within $H_k \cdot \frac{2\alpha k^2}{(1 - \alpha)k + 1}$ iterations. Indeed, we show that each iterations, the PAV score of $W$ increases by at least $\frac{(1 - \alpha)k + 1}{2\alpha k^2}$. As the minimum and maximum PAV scores of a committee are $0$ and $H_k$ respectively, this can occur at most $H_k \cdot \frac{2\alpha k^2}{(1 - \alpha)k + 1}$ times. Hence, we obtain the desired bound on the number of iterations.

    To that end, note that when when we make a swap of $c'$ for $c$, it must be the case that $\hat{\Delta}(W, c) > \frac{1}{\alpha k } - \varepsilon$. Using our assumptions on $\hat{\Delta}$ errors, this implies that $\Delta(W, c) \ge \frac{1}{\alpha k} - 2\varepsilon$. By \Cref{lemPotentialIncrease}, we have that
    \[\max_{x \in W} \Delta(W, c, x) \ge \frac{(1 - \alpha)k + 1}{\alpha k^2} - \frac{k + 1}{k} 2 \varepsilon \ge \frac{(1 - \alpha)k + 1}{\alpha k^2} - 4 \varepsilon. \]
    Again, by our assumption on $\hat{\Delta}$ errors,
    \[\max_{x \in W} \hat{\Delta}(W, c, x) \ge \max_{x \in W} \Delta(W, c, x) - \varepsilon \ge \frac{(1 - \alpha)k + 1}{\alpha k^2} - 5 \varepsilon.
    \]
    Finally, for the choice $c'$ that maximizes $\hat{\Delta}(W, c, c')$,
    \[
        \Delta(W, c, c') \ge \hat{\Delta}(W, c, c') - \varepsilon \ge \frac{(1 - \alpha)k + 1}{\alpha k^2} - 6 \varepsilon = \frac{(1 - \alpha)k + 1}{2\alpha k^2},
    \]
    as needed.
\end{proof}

\section{Description and Analysis of  \Cref{alg:ucb-a-pav}}\label{sec:ucb-a-pav-bound}

In this section, we state \Cref{alg:ucb-a-pav} and analyze its complexity.
The worst-case query guarantees are slightly worse than those of \Cref{alg:noisy-a-pav}; however, as we discuss below, there are instances where \Cref{alg:ucb-a-pav} performs better, and this can additionally be seen in the experiments of \Cref{secEmpirical}.

\begin{algorithm}[htb]
    \caption{$(k,t)$-ucb-$\alpha$-PAV}
    \label{alg:ucb-a-pav}
    \begin{algorithmic}[1] 
        \State Choose $W \in \binom{C}{k}$ arbitrarily

        \State $\mathcal{Q} \gets \set{}$ \Comment{List to store queries and responses}
        \State $\ell \gets 576 \cdot \left(\frac{\alpha k^{2}}{(1-\alpha)k+1}\right)^2\log\left( \frac{4608 k^8 m^{k+2}}{\delta} \right)$ \Comment{Constant to be used later}
        \State $L \gets 2H_{k}\left\lceil \frac{m-k}{t-k}\right\rceil \left(\frac{\alpha k^{2}}{(1-\alpha)k+1}\right) \cdot \ell $ \Comment{Constant to be used later}

        \For{$i = 1, 2, \ldots$}

        \State $V_s(W, x) \gets \set{i \suchthat x \in Q_i \text{ and } |Q_i \cap W| \ge s}$
        \Comment{$\forall s \in \set{0} \cup [k],\forall x \notin W, $}
        \State $V_s(W, x, y) \gets \set{i \suchthat \set{x, y} \subseteq Q_i \text{ and } |Q_i \cap W| \ge s}$
        \Comment{$\forall s \in [k],\forall x \notin W, \forall y \in W$}
        \State $\hat{\Delta}_s^+(W, x) \gets \frac{1}{|V_s|} \sum_{i \in V_s(W, x)} \frac{\mathbb{I}[x \in R_i]}{|R_i \cap W| + 1}$ if $V_s(W, x) \ne \emptyset$ else $\infty$ \par \Comment{Upper bound on estimate for $\Delta(W, x)$ using voters queried on at least $s$ candidates of $W$ (along with $x$)}
        \State $\hat{\Delta}_s^-(W, x, y) \gets \frac{1}{|V_s(W, x, y)|} \sum_{i \in V_s(W, x, y)} \frac{\mathbb{I}[x \in R_i \text{ and } y \notin R_i]}{|R_i \cap W| + |W \setminus Q_i| + 1} - \frac{\mathbb{I}[x \notin R_i \text{ and } y \in R_i]}{|R_i \cap W|}$ if $V_s(W, x, y) \ne \emptyset$ else $-\infty$
        \par\Comment{Lower bound on estimate for $\Delta(W, x, y)$ using voters queried on at least $s$ candidates of $W$ (along with $x$ and $y$)}
        \State $err_s(W, x) \gets \sqrt{\frac{2\log\left(\frac{4L(k + 1)m^{k+1}}{\delta}\right)}{|V_s(W, x)|}}$
        \State $err_s(W, x, y) \gets \sqrt{\frac{2\log\left(\frac{4L(k+1)m^{k+1}}{\delta}\right)}{|V_s(W, x, y)|}}$
        \State $\widetilde{\Delta}^+(W, x) \gets \min_{s \in [k]} \hat{\Delta}_s^+(W, x) + err_s(W, x)$ \Comment{Best UCB-style upper bound on $\Delta(W, x)$ given queries}
        \State $\widetilde{\Delta}^-(W, x, y) \gets \max_{s \in [k]} \hat{\Delta}_s^-(W, x, y) - err_s(W, x, y)$
        \Comment{Best UCB-style lower bound on $\Delta(W, x, y)$ given queries}
        \State $c'\gets\argmax_{x\notin W}\widetilde{\Delta}^+(W,x)$

        \If {$\widetilde{\Delta}^+(W,c')<\frac{1}{\alpha k}$}
        \State \Return
        $W$

        \EndIf

        \State $c\gets\argmax_{x\in W}\widetilde{\Delta}^-(W,c',x)$

        \If{$\widetilde{\Delta}^-(W,c',c)\ge \frac{1}{2}\frac{(1-\alpha)k+1}{\alpha k^{2}}$}

        \State $W\gets(W\cup \set{c'})\setminus \set{c}$

        \Else
        \State $A \gets \set{x \in C \suchthat |\set{i \suchthat W \cup \set{x} \subseteq Q_i}| \ge \ell}$ \Comment{Candidates already queried more than $\ell$ times with $W$}
        \State $S \gets C \setminus W \setminus A$ \Comment{Potential candidates to query along with $W$}
        \State Make query $Q_i$ on $W$ and $t - k$ candidates $x \notin S$ with highest $\widetilde{\Delta}^+(W, x)$, breaking ties arbitrarily
        \State Receive response $R_i$ and append $(i, Q_i, R_i)$ to $\mathcal{Q}$
        \EndIf
        \EndFor
    \end{algorithmic}
\end{algorithm}

\begin{restatable}{thm}{ucbAPAVBound} \label{thm:ucb-a-pav-bound}
    For any $m\ge t > k$, with probability $1-\delta$, \Cref{alg:ucb-a-pav} yields a committee satisfying $\alpha$-OAS
    and $\alpha$-EJR after querying at most
    \begin{equation*}1152 H_{k}\left\lceil \frac{m-k}{t-k}\right\rceil \left(\frac{\alpha k^{2}}{(1-\alpha)k+1}\right)^{3}\log\left(\frac{4608 k^8 m^{k+2}}{\delta}\right)
    \end{equation*}
    voters. For any fixed $\delta > 0$, if $\alpha=1$, this leads to a query complexity of $
        \mathcal{O}\left(mk^7\log k\log m\right)$, and if $\alpha<1$, this leads to a query complexity of
    $\mathcal{O}\left(mk^4\log k\log m\right)$.
\end{restatable}

We use \[
    L := 1152 H_{k}\left\lceil \frac{m-k}{t-k}\right\rceil \left(\frac{\alpha k^{2}}{(1-\alpha)k+1}\right)^{3}\log\left(\frac{4608 k^8 m^{k+2}}{\delta}\right)
\]
to denote this upper bound (notice that it is the same as the $L$ from the algorithm).

Importantly, this theorem states that despite the extensions we introduced in \Cref{alg:ucb-a-pav}, it remains theoretically sound: with a sufficient number of samples, it yields a committee satisfying OAS and EJR.

The proof follows a relatively similar structure to \Cref{thm:noisy-a-pav-bound}: we first show that with probability $1 - \delta$, many estimates are sufficiently accurate, and conditioned on this, the algorithm makes progress in terms of PAV score and terminates with a good committee. However, unlike \Cref{thm:noisy-a-pav-bound}, the samples we take are not fresh for each round, so we can not directly apply Hoeffding's inequality in the most straightforward way. Nonetheless, the proof goes through by instead treating the $\Delta$ estimates as Martingales in order to use Azuma's inequality. Due to its additional intricacy, we separate this portion into its own lemma.

\begin{lem}\label{lem:martingale}
    With probability $1 - \delta$, at every step after querying up to $L$ voters,
    \[\Delta(W, x) \le \widetilde{\Delta}^+(W, x) \le
        \Delta(W, x) + 2 err_k(W, x)\]
    and
    \[
        \Delta(W, x, y) - 2 err_k(W, x, y) \le \widetilde{\Delta}^-(W, x, y) \le \Delta(W, x, y)
    \]
    for all committees $W$, $x \notin W$, and $y \in W$.
\end{lem}
\begin{proof}
    We begin by considering estimates of the form $\widetilde{\Delta}^-(W, x, y) \le \Delta(W, x, y)$; the rest of the estimates will follow similar arguments which we discuss later.
    Fix an arbitrary committee $W$, $x \notin W$, $y \in W$, and $s \in [k]$. We define a sequence of random variables $X_0, X_1, \ldots$ where $X_j$ is the unnormalized estimate $|V_s(W, x, y)| \cdot \hat{\Delta}_s^-(W, x, y)$ when $|V_s(W, x, y)| = j$, i.e., when $j$ voters have been queried on $x$, $y$ and at least $s$ candidates of $W$, and $X_0 = 0$. In other words, when the $j\textsuperscript{th}$ voter of $V_s(W, x, y)$ is queried, $X_j$ is $X_{j  - 1}$ plus that $j\textsuperscript{th}$ voters estimate for $\Delta(W, x, y)$, $\frac{\mathbb{I}[x \in R_i \text{ and } y \notin R_i]}{|R_i \cap W| + |W \setminus Q_i| + 1} - \frac{\mathbb{I}[x \notin R_i \text{ and } y \in R_i]}{|R_i \cap W|}$.

    Notice that when the $j\textsuperscript{th}$ voter is queried, regardless of the algorithm's choices of when to make such a query, this is simply a random voter from the population chosen independently of everything else. Hence, if their entire approval set was known, the expectation of their estimate of $\Delta(W, x, y)$ would be exactly $\Delta(W, x, y)$. When only part $W$ intersects the query, we choose a bound that would always upper bounds the true estimate. Therefore, $\mathbb{E}[X_j \mid X_{j - 1}] \ge X_{j + 1} + \Delta(W, x, y)$.

    Let $Y_0, Y_1, \ldots$ be the additive errors of $X_j$ from the true $\Delta(W, x, y)$, that is, $Y_j = X_j - j \cdot \Delta(W, x, y)$. The key observation we will make is that the sequence $Y_0, Y_1, Y_2, \ldots$ is, in fact, a submartingale. Indeed, since $Y_j = X_j - j \cdot \Delta(W, x, y)$ and $Y_{j - 1} = X_{j - 1} - (j - 1) \cdot \Delta(W, x, y)$, we have $\mathbb{E}[Y_j \mid Y_{j - 1}] \ge 0$.

    Additionally, note that an individual voter's $\Delta$ estimate is always within $[-1, 1]$, so $X_j - X_{j - 1} \in [-1, 1]$. Using the definition of $Y_j$, we have that this implies $Y_j - Y_{j - 1} \in [-1 - \Delta(W, x, y), 1 - \Delta(W, x, y)]$. Note that this is a range of size $2$, and we can hence use (the asymmetric version of) Azuma's inequality to get that for all $\varepsilon > 0$,
    \[
        \Pr[Y_j \le -\varepsilon] = \Pr[Y_j - Y_0 \le -\varepsilon] \le \exp\left( -\frac{2\varepsilon}{j \cdot 2^2} \right) = \exp\left( -\frac{\varepsilon}{2j} \right).
    \] Using this, we can now analyze the errors. When $|V_s(W, x, y)| = j$ for any such $j$,
    \begin{align*}
        \Pr[\hat{\Delta}_s^-(W, x, y) + err_s(W, x, y) \le \Delta(W, x, y)]
         & = \Pr[\hat{\Delta}_s(W, x, y) - \Delta(W, x, y)  \le -err_s(W, x, y)]                                    \\
         & = \Pr[j \cdot \hat{\Delta}_s(W, x, y) - j \cdot \Delta(W, x, y)  \le -j \cdot err_s(W, x, y)]            \\
         & = \Pr[X_j - j \cdot \Delta(W, x, y)  \le -j \cdot err_s(W, x, y)]                                        \\
         & = \Pr[Y_j \le - j \cdot err_s(W, x, y)]                                                                  \\
         & = \Pr\left[Y_j \le -j \cdot \sqrt{\frac{2\log\left(\frac{4L(k+1)m^{k+1}}{\delta}\right)}{j}}\right]      \\
         & = \Pr\left[Y_j \le - \sqrt{2j\log\left(\frac{4L(k+1)m^{k+1}}{\delta}\right)}\right]                      \\
         & \le \exp\left(- \frac{\left(\sqrt{2j\log\left(\frac{4L(k+1)m^{k+1}}{\delta}\right)}\right)^2}{2j}\right) \\
         & = \frac{\delta}{4L(k+1)m^{k+1}}.
    \end{align*}
    Additionally, note that when $s = k$, this is in fact a martingale (no loose upper bounding is needed), so  this inequality continues to hold in other direction  for $\hat{\Delta}_k^-(W, x, y) - err_k(W, x, y) \ge \Delta(W, x, y)$.
    A symmetric argument shows
    \[\Pr[\hat{\Delta}_s^-(W, x) - err_s(W, x) \ge \Delta(W, x)] \le \frac{\delta}{4L(k+1)m^{k+1}}
    \]
    and
    \[
        \Pr[\hat{\Delta}_k^-(W, x) + err_s(W, x) \le \Delta(W, x)] \le \frac{\delta}{4L(k+1)m^{k+1}}.
    \]
    for all $W, x$, and $s$.

    Notice that in the first $L$ queries, the sizes of the $V_s$ sets are trivially upper bounded by $L$. Hence, we can union bound over all at most $L$ sizes, the two choices of either upper and lower bounds, two choices of either $\Delta(W, x, y)$ or $\Delta(W, x)$ the at most $k + 1$ choices of $s$, and at most $m^{k + 1}$ choices of $W$, $x$, and $y$ (we are choosing $m + 1$ candidates with two being special, so clearly at most choosing a sequence of $k + 1$ candidates with repeats). This leads to at most $4L(k + 1)m^{k + 1}$ possible bad events. Hence, with probability $1 - \delta$, none of these bad events happen. Conditioned on this,
    we have that
    \[
        \widetilde{\Delta}^+(W, x) = \min_{s \in [k]} \hat{\Delta}_s^+(W, x) + err_s(W, x) \ge \Delta(W, x)\]
    and
    \[
        \widetilde{\Delta}^-(W, x, y) = \max_{s \in [k]} \hat{\Delta}_s^-(W, x) - err_s(W, x) \ge \Delta(W, x).\]
    In addition, using the bounds on $\hat{\Delta}_k$, we have
    \[
        \widetilde{\Delta}^+(W, x) = \min_{s \in \set{0} \cup [k]} \hat{\Delta}_s^+(W, x) + err_s(W, x) \le \hat{\Delta}_k^+(W, x) + err_k(W, x) \le \Delta(W, x) + 2err_k(W, x)
    \]
    and
    \[
        \widetilde{\Delta}^-(W, x, y) = \max_{s \in  \cup [k]} \hat{\Delta}_s^-(W, x, y) - err_s(W, x, y) \ge \hat{\Delta}_k^-(W, x, y) - err_k(W, x, y) \le \Delta(W, x, y) - 2err_k(W, x, y).
    \]
    Hence, the desired bounds are satisfied.
\end{proof}

We are now ready to prove the theorem.
\begin{proof}[Proof of \Cref{thm:ucb-a-pav-bound}]
    We condition on the event that the estimates after at most $L$ voters are all accurate as in \Cref{lem:martingale}. Let $\ell := 576 \cdot \left(\frac{\alpha k^{2}}{(1-\alpha)k+1}\right)^2\log\left( \frac{4608 k^8 m^{k+2}}{\delta} \right)$ as defined in the algorithm. The technical portion of this proof is to show that for any committee $W$, after at most $\ceil*{\frac{m - k}{t - k}} \cdot \ell$ queries, the algorithm either makes a swap or terminates. Notice that when a swap is made, assuming the estimate is accurate, the PAV score increases by $\frac{(1 - \alpha)k + 1}{2\alpha k^2}$. Hence, just as in previous proofs, such a swap can only happen $2H_k \frac{\alpha k^2}{(1 - \alpha)k}$ times. The choice of $L$ implies termination will occur while estimates are still accurate. Hence, at the point that we terminate, $\Delta^*(W) < \frac{1}{\alpha k}$, so the desired properties are satisfied by \Cref{lemLowDeltaImpliesStuff}.

    What remains is to show that after $\ceil*{\frac{m - k}{t - k}} \cdot \ell$ queries with a committee $W$, either a swap is made or we terminate. By our query selection strategy, after this many queries, $W \cup \set{x}$ will be contained in at least $\ell$ queries for all $x \notin W$. This implies that $|V_k(W, x)| \ge \ell$ and $|V_k(W, x, y)| \ge \ell$ for all such $x$ and $y$. We will later show that when this happens, $err_k(W, x)$ and $err_k(W, x, y)$ are upper bounded by $\varepsilon := \frac{1}{12} \frac{(1 - \alpha)k + 1}{\alpha k^2}$ for all $x$ and $y$. Once this upper bound of $\varepsilon$ has been shown, the proof is very similar of a swap or termination is similar to \Cref{thm:noisy-a-pav-bound}. If $\widetilde{\Delta}^+(W, c') \ge \frac{1}{\alpha k}$, we will certainly terminate. Otherwise, if $\widetilde{\Delta}^+(W, c') < \frac{1}{\alpha k}$, this means $\Delta(W, c') < \frac{1}{\alpha k} - 2 \varepsilon$. Hence, there is a candidate $x$ such that \[
        \Delta(W, c', x) \ge \frac{(1 - \alpha)k + 1}{\alpha k^2} - \frac{k + 1}{k} \cdot 2\varepsilon \ge 12 \varepsilon - 4 \varepsilon = 8 \varepsilon.
    \]
    For such an $x$, \[\widetilde{\Delta}^-(W, c', x) \ge \Delta(W, c', x) - 2\varepsilon \ge 6 \varepsilon = \frac{1}{2} \frac{(1 - \alpha) k + 1}{\alpha k^2}.\] Hence, the swap if condition must pass and a swap will be made.

    Finally, let us show the necessary bound on $err_k$. More formally, we must show
    \[
        \sqrt{\frac{2\log\left(\frac{4L(k + 1)m^{k+1}}{\delta}\right)}{\ell}} \le \varepsilon.
    \]
    Observing that $\ell = \frac{2}{\varepsilon^2} \cdot 2 \log\left( \frac{4608 k^8 m^{k+2}}{\delta} \right)$, it is sufficient to show that \[\log\left(\frac{2L(k + 1)m^{k+1}}{\delta}\right) \le 2 \log\left( \frac{4608 k^8 m^{k+2}}{\delta} \right)\]
    To that end, we have
    \begin{align*}
        \log\left(\frac{4L(k + 1)m^{k+1}}{\delta}\right)
         & \le \log\left(\frac{8Lkm^{k+1}}{\delta}\right)                                                                                                                                                     \\
         & = \log\left(\frac{8 \left( 2H_{k}\left\lceil \frac{m-k}{t-k}\right\rceil \left(\frac{\alpha k^{2}}{(1-\alpha)k+1}\right) \cdot \ell \right) km^{k+1}}{\delta}\right)                               \\
         & \le \log\left(\frac{16 \left( k \cdot m \cdot k^2 \cdot \ell \right) km^{k+1}}{\delta}\right)                                                                                                      \\
         & = \log\left(\frac{16 k^4 m^{k+2} \cdot \ell}{\delta}\right)                                                                                                                                        \\
         & \le \log\left(\frac{16 k^4 m^{k+2} \cdot \frac{2}{\varepsilon^2} \cdot \left( 2 \log\left( \frac{4608 k^8 m^{k+2}}{\delta} \right) \right)}{\delta}\right) \tag{$\frac{1}{\varepsilon} \le 12k^2$} \\
         & \le \log\left(\frac{4608 k^8 m^{k+2} \cdot \left( 2 \log\left( \frac{4608 k^8 m^{k+2}}{\delta} \right) \right)}{\delta}\right)                                                                     \\
         & = \log\left( \frac{4608 k^8 m^{k+2}}{\delta} \right) + \log\left( 2 \log\left( \frac{4608 k^8 m^{k+2}}{\delta} \right) \right)                                                                     \\
         & \le \log\left( \frac{4608 k^8 m^{k+2}}{\delta} \right) + \log\left( \frac{4608 k^8 m^{k+2}}{\delta} \right) \tag{$\log(2a) \le a$ for all $a \in \mathbb{R}$}                                      \\
         & = 2 \log\left( \frac{4608 k^8 m^{k+2}}{\delta} \right),
    \end{align*}
    as needed.
\end{proof}

Comparing \Cref{thm:ucb-a-pav-bound} with \Cref{thm:noisy-a-pav-bound}, we see that our upper bound on the query complexity \Cref{alg:ucb-a-pav} is $k$ times worse asymptotically. However, even in the worst case, it is unclear whether these bounds are tight; the difference may instead be due to slack in our analysis.

Beyond the worst case, there are problem instances where \Cref{alg:ucb-a-pav} requires fewer queries than \Cref{alg:noisy-a-pav}. Consider a setting with $k$ ``good'' candidates supported by all voters and $m'>k$ ``bad'' candidates that no one supports. Note that with $\alpha = 1$, to satisfy EJR, all good candidates must be selected. In this instance, \Cref{alg:noisy-a-pav}
will perform $\Theta(\frac{k^{4}m}{t-k}\log(m))$ queries per
swap. Further, since $m'>k$, with probability at least $\frac{1}{2}$,
even a randomly-selected initial committee contains no more than
$\frac{k}{2}$ good candidates, so $\Omega(k)$ swaps are required.
Hence, \Cref{alg:noisy-a-pav} requires $\Omega(\frac{k^{5}m}{t-k}\log(m))$
queries.

In contrast, \Cref{alg:ucb-a-pav} does not discard votes after
swaps. In particular, consider the estimate $\hat{\Delta}_0^+(W,c)$ for a bad candidate $c$ that
uses all voters that voted on $c$ regardless of if they voted on anyone in $W$. Note that it is always $0$ as no voter ever approves of $c$. Hence, $\widetilde{\Delta}^+(W, c) \le err_0(W, c)$. On the other hand, for all good candidates $c$, $\Delta(W, c) \ge 1/k$, so $\widetilde{\Delta}^+(W, c) \ge 1/k$ as well. Hence, once a bad candidate has been queried $\Omega(k^2 \log m)$ times, it will have a worse $\widetilde{\Delta}^+$ when compared to any good candidate. In addition, only $\Omega(k^2 \log m)$ queries are needed for a good candidate's error term to be small enough to ensure a swap (fewer for the earlier swaps). Hence, at most $O(mk^2 \log m)$ queries are needed for \Cref{alg:ucb-a-pav} to terminate.


In summary, despite being slightly worse in terms of worst-case analysis, there is evidence that \Cref{alg:ucb-a-pav} may work better in practice, an intuition confirmed through experimental comparison in \Cref{secEmpirical}.

\section{Details on Experiments}\label{app:experiments}
The data for each conversation consists of an
$L \times m$ matrix, with component $v_{ij}\in\{\texttt{agree},\texttt{disagree},\texttt{missing}\}$
representing the vote of participant $i$ on comment $j$.

\paragraph{Polis Datasets.}
See \Cref{tbl:polis_stats} for the sizes of the Polis dataset.
\begin{table}
    \caption{Polis datasets statistics: Number of queried voters $L$, number of comments $m$, comments per query $t$, and fraction of comments $t/m$ voted on by each voter.}
    \centering
    \begin{tabular}{lllr}
        \toprule
        $L$  & $m$  & $t$ & $t/m$ \\
        \midrule
        162  & 31   & 20  & 0.65  \\
        1000 & 1719 & 20  & 0.01  \\
        87   & 39   & 20  & 0.51  \\
        353  & 231  & 20  & 0.09  \\
        340  & 209  & 20  & 0.10  \\
        94   & 40   & 20  & 0.50  \\
        1000 & 114  & 20  & 0.18  \\
        230  & 83   & 20  & 0.24  \\
        258  & 98   & 20  & 0.20  \\
        405  & 94   & 20  & 0.21  \\
        278  & 104  & 20  & 0.19  \\
        1000 & 586  & 20  & 0.03  \\
        \bottomrule
    \end{tabular}
    \label{tbl:polis_stats}
\end{table}

\paragraph{Reddit Dataset.} We preprocessed this dataset in the same way as the Polis datasets (including matrix factorization to infer missing votes). Although the output on Reddit differs from Polis (rankings rather than a subset of the comments), the \emph{input} is similar. We can interpret upvotes as approvals and downvotes as disapprovals, so the data fits well with our experiments. See \Cref{secConclusion} for additional discussion on how our approach applies to social media.

\paragraph{Algorithm Parameters}
In practice, for both \Cref{alg:noisy-a-pav} and \Cref{alg:ucb-a-pav}, we treat $\ell$, the number of times we ask voters about each candidate, as a parameter. In addition, for \Cref{alg:ucb-a-pav},
we replace the numerator in the confidence intervals
$err_s$ with a parameter $\theta$.
We assessed $\ell \in \{4, 6, 8, 12, 18, 30\}$ and $\theta \in \{0.03, 0.05, 0.1, 0.2, 0.5, 1.0\}$ on artificial data (e.g., approval profiles generated by sampling each vote independently). We observed that the algorithms are not sensitive to these parameters and picked $\ell = 6$ and $\theta = 0.05$ since they appeared to yield good results based on visual inspection.

\fi

\end{document}